\documentclass[reqno]{amsart}
\usepackage{IEEEtrantools}
\def\baselinestretch{1}

\usepackage{blindtext}
\usepackage{tasks}
\usepackage{graphicx}
\usepackage{mathrsfs}
\usepackage{amsmath,amssymb,amsthm,mathrsfs,amsfonts,dsfont} 
\usepackage{breqn}
\usepackage{mathtools}
\usepackage{enumitem}   
\usepackage{caption}
\usepackage{subcaption}
\mathtoolsset{showonlyrefs=true}
\newtheorem{theorem}{Theorem}[section]
\newtheorem{lemma}[theorem]{Lemma}

\newtheorem{definition}[theorem]{Definition}

\newtheorem{assumption}[theorem]{Assumption}

\newtheorem{remark}[theorem]{Remark}

\newcommand{\argmin}{\textrm{arg}\min}

\newcommand{\Adyn}{A_{\theta}}
\newcommand{\Bdyn}{B_{\theta}}
\newcommand{\Cdyn}{C_{\theta}}
\newcommand{\Ddyn}{D_{\theta}}

\newcommand{\sstf}{ \textrm{SStF-}\textrm{M}_{k}}
\newcommand{\ssf}{ \textrm{SSF-}\textrm{M}_{k}}
\newcommand{\sstft}{ \textrm{SStF-}\textrm{M}_{2}}
\newcommand{\ssft}{ \textrm{SSF-}\textrm{M}_{2}}
\newcommand{\Real}{\mathbb{R}}
\newcommand{\Natural}{\mathbb{N}}
\newcommand{\Expectation}[1]{\ensuremath{\mathds{E}[#1]}}

\newcommand{\infgen}[1]{\mathcal{L}#1}
\newcommand{\define}{\coloneqq}
\newcommand{\Falgebra}{\ensuremath{\mathcal{F}}}
\newcommand{\Pprob}{\ensuremath{\mathds{P}}}
\newcommand{\ExtInputMap}{\ensuremath{\mathcal{U}}}
\newcommand{\IntInputMap}{\ensuremath{\mathcal{W}}}

\newcommand{\horizontalEqSep}{\noindent\rule{18cm}{0.4pt}
}

\newcommand{\diff}{\mathsf{d}}
\newcommand{\defeq}{\vcentcolon=}
\newcommand{\norm}[1]{\lVert#1\rVert}

\newcommand{\kinf}{$\mathcal{K}_{\infty}$}

\mathtoolsset{showonlyrefs=true}

\begin{document}



\begin{abstract}
	In this work, we derive conditions under which compositional abstractions of networks of stochastic hybrid systems can be constructed using the interconnection topology and joint dissipativity-type properties of subsystems and their abstractions. In the proposed framework, the abstraction, itself a stochastic hybrid system (possibly with a lower dimension), can be used as a substitute of the original system in the controller design process. Moreover, we derive conditions for the construction of abstractions for a  class of stochastic hybrid systems involving nonlinearities satisfying an incremental quadratic inequality. In this work, unlike existing results, the stochastic noises and jumps in the concrete subsystem and its abstraction need not to be the same. We provide examples with numerical simulations to illustrate the effectiveness of the proposed dissipativity-type compositional reasoning for interconnected stochastic hybrid systems.
\end{abstract}
%

\title[From Dissipativity Theory to Compositional Abstractions of Interconnected Stochastic Hybrid Systems]{From Dissipativity Theory to Compositional Abstractions of Interconnected Stochastic Hybrid Systems}

\author{Asad Ullah Awan}
\author{Majid Zamani}
\address{Department of Electrical and Computer Engineering, Technical University of Munich, D-80290 Munich, Germany.}
\email{asad.awan@tum.de, zamani@tum.de}

\maketitle

\section{Introduction}

Abstraction based control synthesis is becoming a promising approach to design controllers for enforcing complex specifications over large interconnected control systems in a reliable and cost effective way. 
In this approach, one synthesizes a controller to enforce the complex specifications over the abstraction instead of the original (concrete) system, and refines the controller (using a so-called {\it interface map}) to that of  the concrete system. Since the error between the output of the concrete system and that of its abstraction is quantified, one can ensure that the concrete system also satisfies the specifications (within a priori known error bounds).

Constructing abstractions for a complex system when viewed monolithically is a challenging task in itself. One approach to deal with this is to leverage the fact that many large-scale complex systems can be regarded as interconnected systems consisting of smaller {\it subsystems}. This motivates a {\it compositional} approach for the construction of the abstractions wherein abstractions of the concrete system can be constructed by using the abstractions of smaller subsystems.


Recently, there have been several results on the compositional construction of (in)finite abstractions of deterministic control systems including \cite{pola2016symbolic}, \cite{tazaki2008bisimilar}, \cite{7496809}, and of a class of stochastic hybrid systems \cite{zamani2015approximations}. These results employ a small-gain type condition for the compositional construction of abstractions. However, as shown in \cite{das2004some}, this type of condition is a function of the size of the network and can be violated as the number of subsystems grows. Recently in \cite{7857702}, a compositional framework for the construction of infinite abstractions of networks of control systems has been proposed using dissipativity theory. In this result a notion of storage function is proposed which describes joint dissipativity properties of control systems and their abstractions. This notion is used to derive compositional conditions under which a network of abstractions approximate a network of the concrete subsystems. Those conditions can be independent of the number of the subsystems under some properties for the interconnection topologies. 

In this work, we extend this approach to a class of stochastic hybrid systems, namely, jump-diffusions.  Stochastic hybrid systems are a general class of  systems consisting of continuous and discrete dynamics subject to probabilistic noise and events. In jump-diffusions, the continuous dynamics are modelled by stochastic differential equations and switches are modelled as Poisson processes. 
We introduce a notion of so-called stochastic storage functions describing joint dissipativity properties of stochastic hybrid subsystems and their abstractions. Given a network of stochastic hybrid subsystems and the stochastic storage functions between subsystems and their abstractions, we derive conditions based on the interconnection topology, guaranteeing that a network of abstractions quantitatively approximate the network of concrete subsystems. For a class of stochastic hybrid subsystems and using the incremental quadratic inequality of the nonlinearity, we derive a set of matrix (in)equalities facilitating the construction of their abstractions together with the corresponding stochastic storage functions. We illustrate the effectiveness of the proposed results in two examples in which compositionality conditions are satisfied independent of the number or gains of the subsystems.

\subsection{Related work}
Compositional abstraction for (deterministic) interconnected control systems using dissipativity was introduced in \cite{7857702}. In a preliminary version of this paper, which appeared in \cite{awanIFAC17}, this technique was extended to a class of stochastic hybrid systems. In both works, the joint dissipativity properties are defined with respect to a static map whose input is the (internal) inputs and outputs of the subsystems and their abstractions. In contrast to this, in this paper we employ a dynamic map based on a similar notion introduced in \cite{tippett2011dissipativity}. This allows for a broader class of (stochastic hybrid) subsystems for which one can find (stochastic) storage functions between them and their abstractions (cf. the second case study). Furthermore, in this work we derive constructive conditions for computing abstractions for a class of stochastic hybrid systems by considering nonlinearities which are more general than the ones considered in \cite{7857702} and \cite{awanIFAC17}. 

Compositional abstractions for jump-diffusions are also introduced in \cite{zamani2015approximations}. However, in \cite{zamani2015approximations} it is assumed that the stochastic noises in a subsystem and its abstraction are the same. This assumption is not realistic in practice, as it requires access to the realization of the noises in the original subsystem in order to refine the constructed controllers over the abstractions to the original subsystems. On the other hand, in this paper concrete subsystems and their abstractions do not share the same stochastic noises. In addition, the results in \cite{zamani2015approximations} use small-gain type conditions for the main compostionality result whereas the proposed approach here uses dissipativity-type conditions which can potentially provide scale-free results under some properties over the interconnection topologies. Although the results in \cite{zamani2015approximations} derive conditions for constructing abstractions for just linear jump-diffusions, here we provide constructive conditions for a class of nonlinear jump-diffusions.

 
\section{Stochastic Hybrid Systems}	
\subsection{Notation}
	The sets of non-negative integer and real numbers are denoted by $\Natural$ and $\Real$, respectively. Those symbols are footnoted with subscripts to restrict them in the usual way, e.g. $\Real_{>0}$ denotes the positive real numbers. The symbol $\Real^{n\times m}$ denotes the vector space of real matrices with $n$ rows and $m$ columns. The symbols $\vec{1}_{n}, \vec{0}_n, I_n, 0_{n \times m}$ denote the vector with all its elements to be one, the zero vector, identity and zero matrices in $\Real^n, \Real^n, \Real^{n \times n}, \Real^{n \times m}$, respectively. For $a,b \in \Real$ with $a \leq b$, the closed, open, and half-open intervals in $\Real$ are denoted by $[a,b]$, $]a,b[$, $[a,b[$, and $]a,b]$, respectively. For $a, b \in \Natural$ and $a \leq b$, we use $[a;b]$, $]a;b[$, $[a;b[$, and $]a;b]$ to denote the corresponding intervals in $\Natural$. Given $N \in \Natural_{\geq 1}$, vectors $x_i \in \Real^{n_i}, n_i \in \Natural_{\geq 1}$ and $i \in [1;N]$, we use $x = [x_1;\ldots;x_N]$ to denote the concatenated vector in $\Real^n$ with $n=\sum^N_{i=1} n_i$. Similarly, we use $X = [X_1;\ldots;X_N]$ to denote the matrix in $\Real^{n\times m}$ with $n = \sum^N_{i=1} n_i$, given $N \in \Natural_{\geq 1}$, matrices $X_i \in \Real^{n_i \times m}, n_i \in \Natural_{\geq 1}$, and $i \in [1;N]$. Given a vector $x \in \Real^{n}$, we denote by $\norm{x}$ the Euclidean norm of $x$. Given a matrix $M = \{m_{ij}\} \in \Real^{n\times m}$, we denote by $\norm{M}$ the induced 2 norm of $M$, and the trace of $M$ by $\mathsf{Tr}(M)$, where $\mathsf{Tr}(P) = \sum^n_{i=1}p_{ii}$ for any $P = \{p_{ij}\} \in \Real^{n\times n}$.
	 Given matrices $M_1,\dots,M_n$, the notation $\mathsf{diag}(M_1,\ldots,M_n)$ represents a block diagonal matrix with diagonal matrix entries $M_1,\ldots,M_n$. Given a symmetric matrix $A$, $\lambda_{\min}(A)$ and $\lambda_{\max}(A)$ denote the minimum and maximum eigenvalues of $A$, respectively.
 Given a function $f : \Real_{\geq 0} \rightarrow \Real^n$, the (essential) supremum of $f$ is denoted by $\norm{f}_{\infty} \coloneqq$  (ess)sup$\{\norm{f(t)}, \ t \geq 0\}$. Measurability throughout this paper refers to Borel measurability. A continuous function $\gamma: \Real_{\geq 0} \rightarrow \Real_{\geq 0}$, is said to belong to class $\mathcal{K}$ if it is strictly increasing and $\gamma(0) = 0$; $\gamma$ is said to belong to class $\mathcal{K}_{\infty}$ if $\gamma \in \mathcal{K}$ and $\gamma(r) \rightarrow \infty$ as $r \rightarrow \infty$. A continuous function $\beta: \Real_{\geq 0} \times \Real_{\geq 0} \rightarrow \Real_{\geq 0}$ is said to belong to class $\mathcal{KL}$ if, for each fixed $t$, the map $\beta(r,t)$ belongs to class $\mathcal{K}$ with respect to $r$, and for each fixed nonzero $r$, the map $\beta(r,t)$ is decreasing with respect to $t$ and $\beta(r,t) \rightarrow 0$ as $t \rightarrow \infty$.

\subsection{Stochastic hybrid systems}		
	Let $(\Omega, \Falgebra, \mathds{P})$ denote a probability space endowed with a filtration $\mathds{F} = (\mathcal{F}_s)_{s\geq0}$ satisfying the usual conditions of completeness and right continuity. The expected value of a measurable function $g(X)$, where $X$ is a random variable defined on a probability space ($\Omega, \Falgebra, \Pprob$), is defined by the Lebesgue integral $\Expectation{g(X)} \define \int_{\Omega} g(X(\omega))\diff\mathds{P}(\omega)$, where $\omega \in \Omega$.  Let $(W_s)_{s\geq0}$ be a $\mathsf{b}$-dimensional $\mathds{F}$-Brownian motion and $(P_s)_{s\geq 0}$ be an $\mathsf{r}$-dimensional $\mathds{F}$-Poisson process. We assume that the Poisson process and Brownian motion are independent of each other. The Poisson process $P_s = [P_s^1;\ldots;P_s^{\mathsf{r}} ]$ models $\mathsf{r}$ kinds of events whose occurrences are assumed to be independent of each other. 
\begin{definition} 
	\label{def:shs}
	The class of stochastic hybrid systems studied in this paper is a tuple  $$\Sigma = (\Real^n, \Real^m, \Real^p, \mathcal{U}, \mathcal{W}, f, \sigma, \rho,\Real^{q_1}, \Real^{q_2},h_1,h_2),$$ where
	\begin{itemize}
		
		\item $\Real^n$, $\Real^m$, $\Real^p$, $\Real^{q_1}$, and $\Real^{q_2}$ are the state, external input, internal input, external output, and internal output spaces, respectively;
		\item $\ExtInputMap$ and $\IntInputMap$ are subsets of sets of all $\mathds{F}$-progressively measurable processes taking values in $\Real^m$ and $\Real^p$, respectively;
		\item $f:\Real^n \times \Real^m \times \Real^p \rightarrow \Real^n $  is the drift term which is globally Lipschitz continuous: there exist Lipschitz constants $L_x, L_u, L_w \in \Real_{\geq 0}$ such that $\norm{f(x,u,w) - f(x',u',w')} \leq L_x\norm{x-x'} + L_u\norm{u - u'} + L_w\norm{w-w'}$ for all $x,x'  \in \Real^n$, all $u,u' \in \Real^m$, and all $w, w'\in \Real^p$; 
		\item $\sigma: \Real^n \rightarrow \Real^{n\times \mathsf{b}}$ is the diffusion term which is globally Lipschitz continuous with the Lipschitz constant $L_{\sigma}$;
		\item $\rho: \Real^n \rightarrow \Real^{n \times \mathsf{r}}$ is the reset term which is globally Lipschitz continuous with the Lipschitz constant $L_{\rho}$;
		\item $h_1: \Real^n \rightarrow \Real^{q_1}$ is the external output map;
		\item $h_2: \Real^n \rightarrow \Real^{q_2}$ is the internal output map.
	\end{itemize}
\end{definition}
A stochastic hybrid system $\Sigma$ satisfies
\begin{IEEEeqnarray}{c}	
\Sigma:\left\{ \begin{IEEEeqnarraybox}[\relax][c]{rCl}
\!\diff \xi(t) \!&=&\! f(\xi(t), \upsilon(t), \omega(t))\diff t \!+\! \sigma(\xi(t))\diff W_t \!+\!   \rho(\xi(t)) \diff P_t, \\
\!\zeta_1(t) \!&=&\! h_1(\xi(t)),\\
\!\zeta_2(t) \!&=&\! h_2(\xi(t)),
\end{IEEEeqnarraybox}\right.
\label{eq:sde_big}
\end{IEEEeqnarray}
\normalsize
$\Pprob$-almost surely ($\Pprob$-a.s.) for any $\upsilon \in \ExtInputMap$ and any $\omega \in \IntInputMap$, where stochastic process $\xi: \Omega \times \Real_{\geq 0} \rightarrow \Real^n$ is called a {\it solution process} of $\Sigma$, stochastic process $\zeta_1 : \Omega \times \Real_{\geq 0} \rightarrow \Real^{q_1}$ is called an external output trajectory of $\Sigma$, and  stochastic process $\zeta_2: \Omega \times \Real_{\geq 0} \rightarrow \Real^{q_2}$ is called an internal output trajectory of $\Sigma$. 
We also write $\xi_{a\upsilon\omega}(t)$ to denote the value of the solution process at time $t \in \Real_{\geq 0}$ under input trajectories $\upsilon$ and $\omega$ from initial condition $\xi_{a\upsilon\omega}(0) = a$ $\Pprob$-a.s., where $a$ is a random variable that is $\mathcal{F}_0$-measurable. We denote by $\zeta_{1_{a\upsilon\omega}}$ and $\zeta_{2_{a\upsilon\omega}}$ the external and internal output trajectories corresponding to solution process $\xi_{a\upsilon\omega}$.  Here, we assume that the Poisson processes $P^i_s$, for any $i \in [1;\mathsf{r}]$, have the rates $\lambda_i$. We emphasize that the postulated assumptions on $f, \sigma$, and $\rho$ ensure existence, uniqueness, and strong Markov property of the solution process \cite{oksendal2005applied}. 
	\begin{remark}
	If the stochastic hybrid system $\Sigma$ does not have internal inputs and outputs, the system defined in Definition \ref{def:shs} reduces to $\Sigma = (\Real^n, \Real^m, \mathcal{U}, f, \sigma, \rho, \Real^{q}, h)$, where $f:\Real^n \times \Real^m \rightarrow \Real^n$. Correspondingly, equation \eqref{eq:sde_big} describing the evolution of solution processes reduces to:
	\begin{IEEEeqnarray}{c}
	\label{eq:simple_sde}
	\Sigma: \left\{ \begin{IEEEeqnarraybox}[\relax][c]{rCl}
	\diff \xi(t) \!&=&\! f(\xi(t), \upsilon(t))\diff t \!+\! \sigma(\xi(t))\diff W_t \!+\!   \rho(\xi(t)) \diff P_t, \\
	\zeta(t) \!&=&\! h(\xi(t)) .
	\end{IEEEeqnarraybox}\right.
	\end{IEEEeqnarray}
	\end{remark} 
	We use the notion of stochastic hybrid system as in \eqref{eq:simple_sde} later to refer to interconnected systems.

\section{Stochastic Storage Function}

\label{sec:storage_static}
In this section, we introduce a notion of so-called stochastic storage functions, adapted from the notion of storage functions from dissipativity theory \cite{arcak2016networks}. Before introducing the notion of stochastic storage functions, we introduce a linear control system which is given by: 
\begin{align}
\label{eq:LCS}
\dot{\xi}_{\theta}(t) &= A_{\theta}\xi_{\theta}(t) + B_{\theta}\upsilon_{\theta}(t)\\
\zeta_{\theta}(t) &= C_{\theta}\xi_{\theta}(t) + D_{\theta} \upsilon_{\theta}(t),
\end{align}
	where $A_{\theta} \in \Real^{l_{\theta}\times l_{\theta} }, B_{\theta} \in \Real^{l_{\theta}\times m_{\theta}}, C_{\theta} \in \Real^{ q_{\theta}\times  l_{\theta}}$, and 
$D_{\theta} \in \Real^{q_{\theta} \times m_{\theta}}$, where $B_{\theta}$, and $D_{\theta}$ have the conformal partitions 
\begin{align}
\label{eq:Ddyn}
B_{\theta} = \begin{bmatrix} B_1 & B_2\end{bmatrix},~D_{\theta} = \begin{bmatrix} D_1 & D_2\end{bmatrix},
\end{align}
respectively. These conformal partitions will be used later in the paper. We use the tuple $\Sigma_{\theta} = (A_{\theta},B_{\theta}, C_{\theta}, D_{\theta})$ to represent such a linear control system.
Now we define the infinitesimal generator of a stochastic process which will be used later to define a notion of stochastic storage functions.
\begin{definition}
	Consider two stochastic hybrid systems $\Sigma = (\Real^n, \Real^m, \Real^p, \mathcal{U}, \mathcal{W}, f, \sigma, \rho, \Real^{q_1}, \Real^{q_2}, h_1,h_2)$ and $\hat{\Sigma} = (\Real^{\hat{n}}, \Real^{\hat{m}}, \Real^{\hat{p}}, \mathcal{\hat{U}}, \mathcal{\hat{W}}, \hat{f}, \hat{\sigma}, \hat{\rho}, \Real^{q_1}, \Real^{\hat{q}_2}, \hat{h}_1,\hat{h}_2)$ with solution processes $\xi$ and $\hat{\xi}$, respectively. Consider a linear control system $\Sigma_{\theta} = (A_{\theta},B_{\theta}, C_{\theta}, D_{\theta})$ satisfying \eqref{eq:LCS} with state trajectory $\xi_{\theta}$. Consider a twice continuously differentiable function $V:\Real^n \times \Real^{\hat{n}} \times \Real^{l_{\theta}} \rightarrow \Real_{\geq 0}$. The infinitesimal generator of the stochastic process $\Xi = [\xi;\hat{\xi};\xi_{\theta}]$, denoted by $\infgen$, acting on function $V$ is defined as \cite{oksendal2005applied}:
	\small
	\begin{align}
	\label{eq:infgen}
	\infgen{V}(x,\hat{x},\theta)  &\define 
	\begin{bmatrix} \partial_x V & \partial_{\hat{x}} V & \partial_{\theta}V\end{bmatrix} \begin{bmatrix} f(x,u,w) \\ \hat{f}(\hat{x},\hat{u},\hat{w}) \\
	A_{\theta}\theta + B_{\theta}u_{\theta} \end{bmatrix}    	\nonumber + \frac{1}{2}\mathsf{Tr}\left( \sigma(x)  \sigma^T(x)  \partial_{x,x}V \right)
 + \frac{1}{2}\mathsf{Tr}\left( \hat{\sigma}(\hat x)\hat{\sigma}^T(\hat{x}) \partial_{\hat{x},\hat{x}}V\right) \nonumber \\ &\quad + \sum_{j=1}^{\mathsf{r}} \lambda_j(V(x + \rho(x)\mathsf{e}_j^{\mathsf{r}}, \hat{x} ) - V(x,\hat{x}))   +
	\sum_{j=1}^{\hat{\mathsf{r}}}\hat \lambda_j(V(x,\hat x + \hat{\rho}(\hat{x})\mathsf{e}_j^{\hat{\mathsf{r}}}) - V(x,\hat x)),
	\end{align}
		\normalsize
	where $\mathsf{e}_j^{\mathsf{r}}$ denotes an $\mathsf{r}$-dimensional vector with 1 on the $j$-th entry and 0 elsewhere.
\end{definition}
Now we have all the ingredients to introduce a notion of stochastic storage functions.
\begin{definition}\label{d:stf}
	Consider two stochastic hybrid systems $\Sigma = (\Real^n, \Real^m, \Real^p, \mathcal{U}, \mathcal{W}, f, \sigma, \rho, \Real^{q_1}, \Real^{q_2}, h_1,h_2)$ and $\hat{\Sigma} = (\Real^{\hat{n}}, \Real^{\hat{m}}, \Real^{\hat{p}},  \mathcal{\hat{U}}, \mathcal{\hat{W}}, \hat{f}, \hat{\sigma}, \hat{\rho}, \Real^{q_1}, \Real^{\hat{q}_2}, \hat{h}_1,\hat{h}_2)$ with the same external output space dimension and let $\Sigma_{\theta} = 
	(A_{\theta}, B_{\theta}, C_{\theta}, D_{\theta})$ be a linear control system as in \eqref{eq:LCS}. 
 A twice continuously differentiable function $V:\Real^n \times \Real^{\hat{n}} \times \Real^{l_{\theta}} \rightarrow \Real_{\geq 0}$ is called a stochastic storage function from $\hat{\Sigma}$ to $\Sigma$, with respect to $\Sigma_{\theta}$, in the k-{\normalfont th} moment ($\sstf$), where $k \geq 1$, if it has polynomial growth rate and there exist  convex functions  $\alpha, \eta \in \mathcal{K}_{\infty}$, concave function $\psi_{\mathsf{\mathsf{ext}}} \in \mathcal{K}_{\infty} \cup \{0\} $, some constant $\mathsf{c} \in \Real_{\geq 0}$, some matrices $W, \hat{W},$ and $H$, and some symmetric matrix $X$ of appropriate dimension 
 such that
 	\begin{align}
 \label{ineq:DXnecessary}
 D_2^T X D_2 \preceq 0,
 \end{align} 
 where $D_2$ is given in \eqref{eq:Ddyn}, and $\forall x \in \Real^n$, $\forall\hat{x} \in \Real^{\hat{n}}$, and $\forall\theta \in \Real^{l_{\theta}}$ one has 
	\begin{align}
	\label{in:defV}
	\alpha(\norm{h_1(x) - \hat{h}_1(\hat{x})}^k) \leq V(x,\hat{x},\theta),
	\end{align}
	and $\forall \hat{u} \in \Real^{\hat{m}} \ \exists u \in \Real^m$, such that $\forall \hat{w} \in \Real^{\hat{p}} \ \forall w \in \Real^p$, one obtains
	\begin{align}
	\label{ineq:defDiss}
	 \infgen V(x,\hat{x},\theta) & \leq-\eta(V(x,\hat{x}, \theta)) + \psi_{\mathsf{\mathsf{ext}}}(\norm{\hat{u}}^k)     + z^T X 
	 z +\mathsf{c},
	\end{align}
	where $z = \Cdyn\theta + \Ddyn u_{\theta}$
	and 
	$$u_{\theta} = 
	\begin{bmatrix}
	Ww-\hat{W}\hat{w} \\
	h_2(x) - H\hat{h}_2(\hat{x})\end{bmatrix}.$$ 
	\end{definition}
	We use notation $\hat{\Sigma} \preceq \Sigma$ if there exists an $\sstf$ $V$ from $\hat{\Sigma}$ to $\Sigma$. The stochastic hybrid system $\hat{\Sigma}$ (possibly with $\hat{n} < n$) is called an abstraction of $\Sigma$. 
	\begin{remark}
		\label{rem:quadratic}
		If $C_{\theta}$ is the zero matrix, and $D_{\theta}$ is the identity matrix, then the quadratic term in \eqref{ineq:defDiss} reduces to the one in \cite{7857702, awanIFAC17}, with $$z = \begin{bmatrix} Ww - \hat{W}\hat{w} \\ h_2(x) - H\hat{h}_2(\hat{x})\end{bmatrix}.$$
	\end{remark}
	\begin{remark}
		\label{rem:supplyrate}
		Condition \eqref{ineq:DXnecessary} has also appeared in various forms in the literatures as a necessary condition for deriving asymptotic stability from dissipativity properties of a system. See for example \cite{tippett2011dissipativity}. 
	\end{remark}
		Now, we recall a slightly adapted version of the notion of stochastic simulation function introduced in \cite{zamani2015approximations}. This notion is appropriate for relating interconnected systems without internal inputs and outputs.
	\begin{definition}
		Let $\Sigma = (\Real^n, \Real^m, \mathcal{U}, f, \sigma, \rho, \Real^{q}, h)$ and $\hat{\Sigma} = (\Real^{\hat{n}}, \Real^{\hat{m}}, \mathcal{\hat{U}}, \hat{f}, \hat{\sigma}, \hat{\rho}, \Real^{q}, \hat{h}$) be two stochastic hybrid systems. A twice continuously differentiable function $V: \Real^n \times \Real^{\hat{n}} \times \Real^{l_{\theta}} \rightarrow \Real_{\geq 0}$ is called a stochastic simulation function from $\hat{\Sigma}$ to $\Sigma$ in the k-{\normalfont th} moment  ($\ssf$),  where $k \geq 1$, if there exist convex functions $\alpha, \eta \in \mathcal{K}_{\infty}$, concave function $\psi_{\mathsf{\mathsf{ext}}} \in \mathcal{K}_{\infty} \cup \{0\}$, and some constant $\mathsf c \in \Real_{\geq 0}$, such that $\forall x \in \Real^n$,  $\forall\hat{x} \in \Real^{\hat{n}}$, and $\forall\theta \in \Real^{l_{\theta}}$, one has
		\begin{align}
		\label{ineq:simfunction1}
		\alpha(\norm{h(x) - \hat{h}(\hat{x})}^k) \leq V(x,\hat{x},\theta),
		\end{align}
		and $\forall \hat{u} \in \Real^{\hat{m}} \ \exists u \in \Real^m$ such that
		\begin{align}
		\label{ineq:simfunction2}
		\infgen V(x,\hat{x},\theta) \leq -\eta(V(x,\hat{x},\theta)) + \psi_{\mathsf{\mathsf{ext}}}(\norm{\hat{u}}^k) + \mathsf c.
		\end{align}
	\end{definition}
	We say that a stochastic hybrid  system $\hat{\Sigma}$ is approximately simulated by a stochastic hybrid system $\Sigma$, denoted by $\hat{\Sigma} \preceq_{AS} \Sigma$, if there exists an $\ssf$ function $V$ from $\hat{\Sigma}$ to $\Sigma$. We call $\hat{\Sigma}$ (possibly with lower dimension $\hat{n} < n$) an abstraction of $\Sigma$.
	The next theorem shows the important of the existence of an $\ssf$ by quantifying the error between the output behaviors of $\Sigma$ and the ones of its abstractions $\hat{\Sigma}$. 
	\begin{theorem}\label{theorem1}
		Let $\Sigma = (\Real^n, \Real^m, \mathcal{U}, f, \sigma, \rho, \Real^{q},h)$ and $\hat{\Sigma} = (\Real^{\hat{n}}, \Real^{\hat{m}}, \mathcal{\hat{U}}, \hat{f}, \hat{\sigma}, \hat{\rho}, \Real^{q}, \hat{h})$ be two stochastic hybrid systems. Suppose $V$ is an $\ssf$ from $\hat{\Sigma}$ to $\Sigma$. Then, there exists a $\mathcal{KL}$ function $\beta$, a $\mathcal{K}_{\infty}$ function $\gamma_{\mathsf{\mathsf{ext}}}$, and some constant $\mathsf{c^{\prime}} \in \Real_{\geq 0}$ such that for any $\hat{\upsilon}\in\hat{\ExtInputMap}$,  any random variable $a$ and $\hat{a}$ that are $\mathcal{F}_0$-measurable, and any $\theta_0 \in \Real^{l_{\theta}} $, there exists $\upsilon \in \ExtInputMap$ such that the following inequality holds for any $t \in \Real_{\geq 0}$:
		\begin{align}
		\label{eq:bound_output}
		\Expectation{\norm{\zeta_{a\upsilon}(t) - \hat{\zeta}_{\hat{a}\hat{\upsilon}}(t)}^k} &\leq \beta(\Expectation{V(a,\hat{a},\theta_0)},t)+ \gamma_{\mathsf{\mathsf{ext}}}(\Expectation{\norm{\hat{\upsilon}}^k_{\infty}}) + \mathsf{c}^{\prime}.
		\end{align}
	\end{theorem}
	\begin{proof}
	The proof is similar to the one of Theorem 3.5 in \cite{zamani2015approximations} and is omitted here due to lack of space. 
	\end{proof}

\section{Interconnected Stochastic Hybrid Systems}
Next definition provides a notion of interconnection for stochastic hybrid susystems investigated in this paper.
\begin{definition}
	Consider $N \in \Natural_{\geq1}$ stochastic hybrid subsystems $$\Sigma_i = (\Real^{n_i}, \Real^{m_i}, \Real^{p_i}, \mathcal{U}_i, \mathcal{W}_i, f_i, \sigma_i, \rho_i, \Real^{q_{1i}}, \Real^{q_{2i}}, h_{1i},h_{2i}),$$ where $i \in [1;N]$, and a static matrix $M$ (the interconnection matrix) of an appropriate dimension defining the coupling of these subsystems. The interconnected stochastic hybrid system $$\Sigma = (\Real^{n}, \Real^{m},\mathcal{U}, f, \sigma, \rho, \Real^{q}, h),$$ denoted by $\mathcal{I}(\Sigma_1,\dots,\Sigma_N)$, follows by $n = \sum^N_{i=1} n_i, m = \sum^{N}_{i=1} m_i, q = \sum^{N}_{i=1}q_{1i}$, and the functions 
	\begin{align}
	f(x,u) &\define [f_1 (x_1,u_1,w_1);\dots;f_N (x_N,u_N,w_N)], \\
	\sigma(x) &\define [\sigma_1(x_1);\dots;\sigma_N(x_N)], \\
	\rho(x) &\define [\rho_1(x_1);\dots;\rho_N(x_N)], \\
	h(x) &\define [h_{11}(x_1);\dots;h_{1N} (x_N)],
	\end{align}
	
	where $u = [u_1;\dots;u_N]$, $x=[x_1;\dots;x_N]$ and with internal variables constrained by 
	\begin{align}
	[w_1;\dots;w_N] = M[h_{21}(x_1);\dots;h_{2N}(x_N)].
	\end{align}
\end{definition}

Assume we are given $N$ stochastic hybrid subsystems $\Sigma_i = (\Real^{n_i}, \Real^{m_i}, \Real^{p_i}, \mathcal{U}_i, \mathcal{W}_i, f_i, \sigma_i, \rho_i, \Real^{q_{1i}}, \Real^{q_{2i}}, h_{1i},h_{2i})$
together with their corresponding abstractions $\hat\Sigma_i = (\Real^{\hat n_i}, \Real^{\hat m_i}, \Real^{\hat p_i}, \hat{\mathcal{U}}_i, \hat{\mathcal{W}}_i, \hat f_i, \hat\sigma_i, \hat\rho_i, \Real^{q_{1i}}, \Real^{\hat q_{2i}}, \hat h_{1i},\hat h_{2i})$ and with $\sstf$ $V_i$ from
$\hat\Sigma_i$ to $\Sigma_i$. We use $\alpha_{i}$, $\eta_i$, $\psi_{i\mathrm{ext}}$, $A_{\theta_i}$, $B_{\theta_i}$, $C_{\theta_i}$, $D_{\theta_i}$, $H_i$, $W_i$, $\hat W_i$, and $X_i$ to denote the corresponding functions, matrices, and their corresponding conformal block partitions appearing in Definition \ref{d:stf}. The next theorem provides a compositional approach on the construction of abstractions of networks of stochastic hybrid systems. 
	\begin{theorem}
	Consider an interconnected system $\Sigma = \mathcal{I}(\Sigma_1,\dots,\Sigma_N)$ induced by $N \in \Natural_{\geq 1}$ stochastic hybrid subsystems $\Sigma_i$ and the interconnection matrix $M$. Suppose each subsystem $\Sigma_i$ admits an abstraction $\hat{\Sigma}_i$ with the corresponding $\sstf$ $V_i$ with respect to $\Sigma_{\theta_i} = (A_{\theta_ i},B_{\theta_i},C_{\theta_i},D_{\theta_i})$, $i \in [1;N]$. Suppose there exists $\mu_i > 0$, $i \in [1;N]$, symmetric matrix $\tilde{Q} \succeq 0$, and matrix $\hat{M}$ of appropriate dimension such that the matrix (in)equalities \eqref{eq:dyninterconnected1} and \eqref{eq:dycinterconnected2} 
	\begin{figure*}[!t]
		\small
		\begin{align} 
	\label{eq:dyninterconnected1}
	\begin{bmatrix} A_D^T\tilde{Q} + \tilde{Q}A_D& \tilde{Q}B_DS\begin{bmatrix} WM \\ I_{\tilde{q}} \end{bmatrix} \\ \begin{bmatrix}WM \\ I_{\tilde{q}} \end{bmatrix}^TS^TB_D^T\tilde{Q} & 0 \end{bmatrix} + \begin{bmatrix}C_D  & D_DS \begin{bmatrix} WM \\ I_{\tilde{q}} \end{bmatrix} \end{bmatrix}^T\begin{bmatrix}\mu_1X_1 & & \\ & \ddots & \\ & & \mu_N X_N \end{bmatrix} \begin{bmatrix}C_D  & D_DS \begin{bmatrix} WM \\ I_{\tilde{q}} \end{bmatrix} \end{bmatrix} &\preceq 0,\\
	\label{eq:dycinterconnected2}
	WMH &= \hat{W}\hat{M},
	\end{align}
		\normalsize
		\horizontalEqSep
	\end{figure*}
	are satisfied, where  $\tilde{q} = \sum^N_{i=1}q_{2i}$, and
	\begin{align}
	W &= \mathsf{diag}(W_1,\ldots,W_N), 
	\hat{W} = \mathsf{diag}(\hat{W}_1,\ldots,\hat{W}_N), 
	H = \mathsf{diag}(H_1,\ldots,H_N), \nonumber \\
	A_D &= \mathsf{diag}(A_{\theta_1},\dots,A_{\theta_N}), B_D = \mathsf{diag}(B_{\theta_1},\dots,B_{\theta_N}), \nonumber \\ C_D &= \mathsf{diag}(C_{\theta_1},\dots,C_{\theta_N}), 
	D_D = \mathsf{diag}(D_{\theta_1},\dots,D_{\theta_N}),
	\end{align}
	and $S$ is the following permutation matrix: 
	\small
	\begin{align}
	S = \begin{bmatrix}
	I_{r_{W_1}} & 0_{r_{W_2}} & \dots & 0_{r_{W_N}} & 0_{r_{H_1}} & 0_{r_{H_2}} &\dots & 0_{r_{H_N}}  \\
	0_{r_{W_1}} & 0_{r_{W_2}} & \dots & 0_{r_{W_N}} & I_{r_{H_1}} & 0_{r_{H_2}}& \dots & 0_{r_{H_N}} \\
	0_{r_{W_1}} & I_{r_{W_2}} & \dots & 0_{r_{W_N}} & 0_{r_{H_1}} & 0_{r_{H_2}}& \dots & 0_{r_{H_N}} \\
	0_{r_{W_1}} & 0_{r_{W_2}} & \dots & 0_{r_{W_N}} & 0_{r_{H_1}} & I_{r_{H_2}}& \dots & 0_{r_{H_N}} \\
	\vdots  & \ & \ddots & \vdots & \vdots & \ & \ddots & \vdots \\
	0_{r_{W_1}} & 0_{r_{W_2}} & \dots & I_{r_{W_N}} & 0_{r_{H_1}} & 0_{r_{H_2}}& \dots & 0_{r_{H_N}} \\
	0_{r_{W_1}} & 0_{r_{W_2}} & \dots & 0_{r_{W_N}} & 0_{r_{H_1}} & 0_{r_{H_2}}& \dots & I_{r_{H_N}}
	\end{bmatrix},
	\end{align}
	\normalsize
	where, for each $i \in [1;N]$, $r_{W_i}$ and $r_{H_i}$ denote the number of rows in $W_i$ and $H_i$, respectively.
	\normalsize
	Then $$V(x,\hat{x},\theta) \define \sum^{N}_{i=1}\mu_i V_i(x_i,\hat{x}_i,\theta_i) + \theta^T \tilde{Q} \theta,$$ where $\theta \defeq [\theta_1;\dots;\theta_N] \in \Real^{l_{\theta}}, l_{\theta} = \sum_{i=1}^{N} l_{\theta_i}$, is an $\ssf$ from the interconnected system $\hat{\Sigma} \define \mathcal{I}(\hat{\Sigma_1},\dots,\hat{\Sigma}_N)$, with the coupling matrix $\hat{M}$, to $\Sigma$.
\end{theorem}
\begin{proof}
	The proof is inspired by that of Theorem 4.2 in \cite{7857702}. First we show that the inequality (\ref{ineq:simfunction1}) holds for some convex $\mathcal{K}_{\infty}$ function $\alpha$. As also argued in the proof of Theorem 4.2 in \cite{zamani2015approximations}, for any $x = [x_1;\dots;x_N] \in \Real^n$, any $\hat{x} = [\hat{x}_1;\dots;\hat{x}_N] \in \Real^{\hat{n}}$, and any $\theta \defeq [\theta_1;\dots;\theta_N] \in \Real^{l_{\theta}}$, one gets:
	\small
	\begin{align*}
	\norm{h(x) - \hat{h}(\hat{x})}^k &\leq  N^{\max\{\frac{k}{2},1\}-1}	\sum^{N}_{i=1}\norm{h_{1i}(x_i) - \hat{h}_{1i}(\hat{x}_i)}^k \nonumber \\ 
	&\leq N^{\max\{\frac{k}{2},1\}-1}\sum^{N}_{i=1}\alpha^{-1}_i(V_i(x_i,\hat{x}_i,\theta_i))  \nonumber \\ &\leq \underline{\alpha}(V(x,\hat{x},\theta)), \nonumber
	\end{align*}
	\normalsize
	for any $k \geq 1$, where $\underline{\alpha}$ is a $\mathcal{K}_{\infty}$ function defined as 
	\begin{align}
	\underline{\alpha}(s) \define \begin{cases}   \max\limits_{\vec{s}  \geq 0} & N^{\max\{\frac{k}{2},1\}-1}\sum^N\limits_{i=1}\alpha^{-1}_i(s_i)\\ \mbox{s.t.} & \mu^T\vec{s} = s,\end{cases}
	\end{align}
	where $\vec{s} = [s_1;\dots;s_N] \in \Real^N$ and $\mu = [\mu_1;\dots;\mu_N]$. The function $\underline{\alpha}$ is a concave function as argued in \cite{zamani2015approximations}. By defining the convex function\footnote{The inverse of a strictly increasing concave (resp. convex) function is a strictly increasing convex (resp. concave) function.} $\alpha(s) = \underline{\alpha}^{-1}(s), \forall s \in \Real_{\geq 0}$, one obtains
	$$\alpha(\norm{h_{1}(x) - \hat{h}_1(\hat{x})}^k ) \leq V(x,\hat{x},\theta),$$
	satisfying inequality \eqref{ineq:simfunction1}. 
	Now we prove the inequality \eqref{ineq:simfunction2}. Consider any $x=[x_1;\dots;x_N] \in \Real^n, \hat{x} = [\hat{x}_1;\dots;\hat{x}_N] \in \Real^{\hat{n}},$ and $\hat{u} = [\hat{u}_1;\dots;\hat{u}_N] \in \Real^{\hat{m}}$. For any $i \in [1;N]$, there exists $u_i \in \Real^{m_i}$, consequently, a vector $u = [u_1;\ldots;u_N] \in \Real^m$, satisfying  (\ref{ineq:defDiss}) for each pair of subsystems $\Sigma_i$ and $\hat{\Sigma}_i$ with the internal inputs given by $[w_1;\dots;w_N] = M[h_{21}(x_1);\dots;h_{2N}(x_N)]$ and $[\hat{w}_1;\dots;\hat{w}_N] = \hat{M}[\hat{h}_{21}(\hat{x}_1);\dots;\hat{h}_{2N}(\hat{x}_N)],$ respectively.  
	The dynamics of $\Sigma_{\theta_i}$, $i \in [1;N]$, can be lumped together into a single auxiliary system as the following:
	\small
	\begin{align}
	\dot{\theta}(t) &= A_D\theta(t) + B_DS	\begin{bmatrix}
	W_1w_1 - \hat{W}_1\hat{w}_1 \\ \vdots \\ W_Nw_N - \hat{W}_N\hat{w}_N \\ h_{21}(x_1) - H_1\hat{h}_{21}(\hat{x}_1) \\ \vdots \\  h_{2N}(x_N) - H_N\hat{h}_{2N}(\hat{x}_N)	\end{bmatrix} \nonumber \\ &= A_D\theta(t) + B_DS\begin{bmatrix}WM \\ I_{\tilde{q}} \end{bmatrix} \begin{bmatrix} h_{21}(x_1) - H_1\hat{h}_{21}(\hat{x}_1) \\ \vdots \\  h_{2N}(x_N) - H_N\hat{h}_{2N}(\hat{x}_N)\end{bmatrix}, \nonumber\\
	z(t) &= C_D\theta(t) + D_DS\begin{bmatrix}
	W_1w_1 - \hat{W}_1\hat{w}_1 \\ \vdots \\ W_Nw_N - \hat{W}_N\hat{w}_N \\ h_{21}(x_1) - H_1\hat{h}_{21}(\hat{x}_1) \\ \vdots \\  h_{2N}(x_N) - H_N\hat{h}_{2N}(\hat{x}_N)	\end{bmatrix} \nonumber \\&= C_D\theta(t) + D_DS\begin{bmatrix}WM \\ I_{\tilde{q}} \end{bmatrix} \begin{bmatrix} h_{21}(x_1) - H_1\hat{h}_{21}(\hat{x}_1) \\ \vdots \\  h_{2N}(x_N) - H_N\hat{h}_{2N}(\hat{x}_N)\end{bmatrix},
	\end{align}
	\normalsize
	where $z = [z_1;\dots;z_N]$. We now consider the infinitesimal generator of the function $V$, and employ the previous auxiliary system and conditions \eqref{eq:dyninterconnected1} and \eqref{eq:dycinterconnected2} to derive the chain of inequalities given in \eqref{ineq:inf_storage},
	\begin{figure*}[!t]
	\small
	\begin{align}\notag
	\infgen V(x,\hat{x},\theta) &= \sum^N_{i=1}\mu_i\infgen V_i(x_i,\hat{x}_i,\theta_i) + \dot{\theta}^T\tilde{Q}\theta + \theta^T\tilde{Q}\dot{\theta}\nonumber 
	\leq \sum^N_{i=1} \mu_i \Bigg(\!-\eta_i(V_i(x_i,\hat{x}_i,\theta_i)) + \psi_{i\mathsf{ext}}(\norm{\hat{u}_i}^k) \nonumber  +  z_i^T 
	X_i z_i + \mathsf c_i\!\Bigg) \!+\! \dot{\theta}^T\tilde{Q}\theta \!+\! \theta^T\tilde{Q}\dot{\theta}\nonumber  \\
	&=  -\sum^N_{i=1}\mu_i \eta_i(V_i(x_i,\hat{x}_i,\theta_i)) + \sum^N_{i=1}\mu_i\psi_{i\mathsf{ext}}(\norm{\hat{u}_i}^k)\nonumber + 
	\begin{bmatrix}
	z_1 \\ \vdots \\ z_N \end{bmatrix}^T 
	\begin{bmatrix} \mu_1X_1 & & \\ & \ddots & \\ & & \mu_NX_N \end{bmatrix}		
	\underbrace{\begin{bmatrix} z_1 \\ \vdots \\ z_N \end{bmatrix}}_{z}
	\nonumber \\ &\quad	+
	\Theta(x,\theta)^T
	\begin{bmatrix} A_D^T\tilde{Q}+\tilde{Q}A_D &  \tilde{Q}B_D S\begin{bmatrix}WM \\ I_{\tilde{q}}\end{bmatrix}
	\\
	\begin{bmatrix}WM \\ I_{\tilde{q}}\end{bmatrix}^TS^TB_D^T\tilde{Q} & 0 
	\end{bmatrix}
	\Theta(x,\theta)
	+ \mathsf c^{\prime}\nonumber \\
	&=  -\sum^N_{i=1}\mu_i \eta_i(V_i(x_i,\hat{x}_i,\theta_i)) +  \sum^N_{i=1}\mu_i\psi_{i\mathsf{ext}}(\norm{\hat{u}_i}^k)\nonumber + \Theta(x,\theta)^T 
		\begin{bmatrix} A_D^T\tilde{Q}+\tilde{Q}A_D &  \tilde{Q}B_D S\begin{bmatrix}WM \\ I_{\tilde{q}}\end{bmatrix}
	\\
	\begin{bmatrix}WM \\ I_{\tilde{q}}\end{bmatrix}^TS^TB_D^T\tilde{Q} & 0 
	\end{bmatrix}
	\Theta(x,\theta) \nonumber \\
	& \quad +
		\Theta(x,\theta)^T
		\begin{bmatrix}
			C_D & D_D S\begin{bmatrix}WM \\ I_{\tilde{q}} \end{bmatrix}
		\end{bmatrix}^T
		\begin{bmatrix} 
			\mu_1X_1 & & \\ & \ddots & \\ & & \mu_NX_N
		\end{bmatrix}	
	\begin{bmatrix}
	C_D & D_D S\begin{bmatrix}WM \\ I_{\tilde{q}} \end{bmatrix}
	\end{bmatrix}
	\Theta(x,\theta) + \mathsf c^{\prime}
	 \\\label{ineq:inf_storage}
	&\leq -\eta(V(x,\hat{x},\theta)) + \psi_{\mathsf{\mathsf{ext}}}(\norm{\hat{u}}^k) + \mathsf c^{\prime},
	\end{align} 
	\horizontalEqSep
	\normalsize	
	\end{figure*}
	where $\mathsf c^{\prime} = \sum_{i=1}^{N} \mu_i \mathsf{c}_i,$
	\small
	\begin{align}
	 \Theta(x,\theta) \define \begin{bmatrix} \theta_1 \\ \vdots \\ \theta_N \\  h_{21}(x_1) - H_1\hat{h}_{21}(\hat{x}_1) \\ \vdots \\  h_{2N}(x_N) - H_N\hat{h}_{2N}(\hat{x}_N) \end{bmatrix},
	\end{align}
	\normalsize
	and the functions $\eta \in \mathcal{K}_{\infty}$ and $\psi_{\mathsf{\mathsf{ext}}} \in \mathcal{K}_{\infty} \cup \{0\}$ are defined as
	\begin{align}
	\eta(s) \coloneqq \begin{cases} \min\limits_{\vec{s}  \geq 0} & \sum^N_{i=1} \mu_i\eta_i(s_i) \\ \mbox{s.t.} & \mu^T\vec{s} = s,\end{cases}
	\end{align}
	\begin{align}
	\psi_{\mathsf{\mathsf{ext}}}(s) \coloneqq \begin{cases} \max\limits_{\vec{s}  \geq 0} & \sum^N_{i=1} \mu_i\psi_{i\mathsf{ext}}(s_i) \\ \mbox{s.t.} & \norm{\vec{s}} \leq s.\end{cases}\end{align}
	It remains to show that $\eta$ is a convex function and $\psi_{\mathsf{\mathsf{ext}}}$ is a concave one. Let us recall that by assumption functions $\eta_i, \   \forall i \in [1;N]$, are convex functions. Thus the function $\eta$ above defines a {\it perturbation function} which is a convex one; see \cite{boyd2004convex} for further details. By similar reasoning, by assumption $\psi_{i\mathsf{ext}}$, $\forall i \in [1;N]$, are concave functions. We conclude that $\psi_{\mathsf{\mathsf{ext}}}$ is a concave function. Hence, we conclude $V$ is an $\ssf$ function from $\hat{\Sigma}$ to $\Sigma$.
\end{proof}
In the next section, we consider a specific class of stochastic hybrid systems $\Sigma$, and a specific candidate $\sstft$ function $V$. We derive conditions under which a given $\hat\Sigma$ is an abstraction of $\Sigma$ and $V$ is an $\sstft$  from $\hat\Sigma$ to $\Sigma$. 
\begin{remark}
If $C_{\theta_i}$ is the zero matrix and $D_{\theta_i}$ is the identity matrix (i.e. $\Sigma_{{\theta}_i}$ is a static map), $\forall i\in [1;N]$, then matrix inequality \eqref{eq:dyninterconnected1} reduces to matrix inequality (15) in \cite[Theorem 7]{awanIFAC17} (which is a stochastic counterpart of matrix inequality (IV.1) in \cite[Theorem 4.2]{7857702}).
\end{remark}
\section{A Class of Stochastic Hybrid Systems}
We consider a specific class of stochastic hybrid systems with the drift, diffusion, reset, and output functions given by 
\begin{align}
\label{eq:JLSS}
\diff \xi(t) &= (A\xi (t) +  B\upsilon (t) + E\varphi(t,F\xi) + D\omega (t))\diff t  \nonumber  +G\diff W_t +  \sum_{i=1}^{\mathsf{r}}R_i\diff P_t^i,\nonumber\\
\zeta_1(t) &= C_1\xi(t),\nonumber\\
\zeta_2(t) &= C_2\xi(t),
\end{align}
where $A \in \Real^{n\times n}, B\in \Real^{n\times m}, D\in \Real^{n\times p}, E \in \Real^{n\times l_k}, F \in \Real^{l_k \times n}, G \in \Real^{n\times 1}, R_i \in \Real^{n}, \forall i \in [1;\mathsf{r}], C_1 \in \Real^{q_1\times n}$, and $C_2 \in \Real^{q_2\times n}$.
The vector $R_i$ and scalar $\lambda_i>0$ (rate of the Poisson process), $\forall i \in [1;\mathsf{r}]$, parametrize the jumps associated with events $i$. 
The time-varying non-linearity is the one considered in \cite{accikmecse2011observers}, which satisfies an incremental quadratic inequality: 
for all $\tilde{M} \in \mathcal{M}$, where $\mathcal{M}$ is the set of symmetric matrices referred to as incremental multiplier matrices, the following incremental quadratic constraint holds for all $t \in \Real_{\geq 0}$, and $k_1, k_2 \in \Real^{l_k}$:
\begin{align}
\begin{bmatrix}k_2 - k_1 \\ \varphi(t,k_2) - \varphi(t,k_1) \end{bmatrix}^T \tilde{M} \begin{bmatrix}k_2 - k_1 \\ \varphi(t,k_2) - \varphi(t,k_1) \end{bmatrix} \geq 0.
\end{align}
To facilitate subsequent analysis, we write matrix $\tilde{M}$ in the following conformal partitioned form
\begin{align}
\tilde{M} = \begin{bmatrix}
M_{11} & M_{12} \\ M_{12}^T & M_{22}
\end{bmatrix}.
\end{align}
We use the tuple
$$\Sigma = (A,B,C_1,C_2,D,E,F, G,\mathsf{R}, \varphi, \lambda),$$
where $\mathsf{R} = \{R_1,\dots,R_{\mathsf{r}}\}$ and $\lambda = \{\lambda_1,\dots,\lambda_{\mathsf r}\}$, to refer to the class of system of the form \eqref{eq:JLSS}. We now consider a specific candidate function and derive conditions under which it is an $\sstft$  from $\hat{\Sigma}$ to $\Sigma$. 
\subsection{Stochastic storage function}
	Here, we consider a candidate $\sstft$ of the form
	\begin{align}
		\label{eq:quadraticstoragefunction}
		V(x,\hat{x},\theta) = (x-P\hat{x})^T\widehat{M}(x-P\hat{x}) + \theta^T\Lambda\theta,
	\end{align}
	where $P$, $\widehat{M}^T= \widehat{M} \succ 0$, and $\Lambda = \Lambda^T \succ 0$ are matrices of appropriate dimensions.  In order to show that $V(x,\hat{x},\theta)$ in \eqref{eq:quadraticstoragefunction} is an $\sstft$ from an abstraction $\hat{\Sigma}$ to the concrete system $\Sigma$, with respect to $\Sigma_{\theta} = (\Adyn, \Bdyn, \Cdyn, \Ddyn)$, where $B_{\theta} = \begin{bmatrix} B_1 & B_2 \end{bmatrix}$ and $D_{\theta} = \begin{bmatrix} D_1 & D_2 \end{bmatrix}$, we require the following assumptions on the concrete system $\Sigma$ and on $\Sigma_{\theta}$.
	
\begin{assumption}
	\label{assum:first}
	Let $\Sigma$ = $(A,B,C_1,C_2,D,E,F,G,\mathsf{R},\varphi, \lambda)$. There exist matrices $\widehat{M} \succ 0$, $K$, $X$, $L_1$, $Z$, $\Lambda$, $\Adyn$, $\Cdyn$, $B_\theta:=[B_1~B_2]$, $D_\theta:=[D_1~D_2]$, and positive constants $\hat{\kappa}$ and $\bar{\kappa}$, such that $$ D_2^TX D_2 \preceq 0,$$and the  (in)equalities given in \eqref{eq:dzw} and \eqref{ineq:bmi} hold,
	\begin{figure*}[!t]
	\footnotesize
	\begin{align}\label{eq:dzw}
	D &= ZW,\\ \label{ineq:bmi}
\begin{bmatrix}
	\Delta\!&\! \widehat{M}Z \!&\! \widehat{M}(BL_1+E)  \!&\!  C_2^TB_2^T\Lambda \\  
	Z^T\widehat{M} \!&\! 0 \!&\! 0 \!&\! B_1^T\Lambda \\ 
	(BL_1+E)^T\widehat{M} \!&\! 0 \!&\! 0 \!&\! 0 \\
	\Lambda B_2 C_2 \!&\! \Lambda B_1 \!&\! 0 \!&\! {\Adyn}^T\Lambda + \Lambda{\Adyn}
	\end{bmatrix} \!&\preceq\! \begin{bmatrix}
	-\hat{\kappa}\widehat{M} +C_2^TD_2^TXD_2C_2 -  F^T M_{11} F \!&\!C_2^TD_2^TXD_1 \!&\! -F^T M_{12} \!&\! C_2^TD_2^TXC_{\theta} \\
	D_1^T X D_2 C_2  \!&\!  D_1^TXD_1 \!&\! 0 \!&\! D_1^TXC_{\theta} \\
	-M_{12}^T F  \!&\! 0 \!&\! -M_{22} \!&\! 0\\
	C_{\theta}^TXD_2C_2 \!&\! C_{\theta}^TXD_1 \!&\! 0 \!&\! {\Cdyn}^TX{\Cdyn} - \bar{\kappa}\Lambda
	\end{bmatrix},
	\end{align}
	\horizontalEqSep
	\end{figure*}
	\normalsize
	where
	\begin{align}
	\Delta &= 
	(A+BK)^T\widehat{M} + \widehat{M}(A+BK).
	\end{align} 
	
\end{assumption}
	An equivalent geometric characterization of \eqref{eq:dzw} is given by the following lemma.
	\begin{lemma}
		Given $D$ and $Z$, the condition \eqref{eq:dzw} is satisfied for some matrix $W$ if and only if
		\begin{align}
		\mathrm{im} \ D \subseteq \mathrm{im} \ Z.
		\end{align}
	\end{lemma}
Now, we provide one of the main results of this section showing under which conditions $V$ in \eqref{eq:quadraticstoragefunction} is an $\sstft$.
\begin{remark}
Remark that when the non-linearity in \eqref{eq:JLSS} reduces to the one described in \cite[Section V]{7857702} and $\Sigma_{\theta}$ is a static map, matrix inequality \eqref{ineq:bmi} reduces to (V.5) in \cite[Theorem 5.5]{7857702}. Note also that in the absence of the non-linearity in \eqref{eq:JLSS}, matrix inequality \eqref{ineq:bmi} is feasible if the pair $(A,B)$ is stabilizable and $A_{\theta}$ is Hurwitz. 
\end{remark}

\begin{theorem}
	Let $\Sigma = (A,B,C_1,C_2,D,E,F,G,\mathsf{R},\varphi, \lambda)$, and $\hat{\Sigma} = (\hat{A},\hat{B},\hat{C}_1,\hat{C}_2,\hat{D},\hat{E},\hat{F},\hat{G},\hat{\mathsf{R}},\varphi,\hat{\lambda})$ with the same external output dimension. Suppose Assumption \ref{assum:first} holds and there exist matrices $P$, $Q$, $H$, $\hat{W}$ and $L_2$ of appropriate dimensions such that: 
	\begin{subequations}
		\label{eq:assumAll}
		\begin{align}
		AP &= P\hat{A} - BQ \label{eq:appabq}\\
		C_1P &= \hat{C}_1 \label{eq:c1pHatc1}\\
		C_2P &= H\hat{C}_2 \label{x12c2px12hHatc2}\\
		FP &= \hat{F} \label{fpHatf}\\
		E &= P\hat{E} + B(L_2 - L_1) \label{epHatEbl1l2}\\
		P\hat{D} &= Z\hat{W} \label{pHatdzHatw}.
		\end{align}
	\end{subequations}
Then, function $V$ defined in \eqref{eq:quadraticstoragefunction} is an $\sstft$ from $\hat{\Sigma}$ to $\Sigma$, with respect to  $\Sigma_{\theta} = (A_{\theta}, B_{\theta}, C_{\theta}, D_{\theta})$.
\end{theorem}
\begin{proof}
		We note that from \eqref{eq:c1pHatc1}, $\forall x \in \Real^n$ and $\forall \hat{x} \in \Real^{\hat{n}}$, we have $\norm{C_1x - \hat{C}_1\hat{x}}^2 = (x-P\hat{x})^TC_1^TC_1(x-P\hat{x})$. It can be readily verified that $\frac{\lambda_{\min}(\widehat{M})}{\lambda_{\max}(C_1^TC_1)}\norm{C_1x - \hat{C}_1\hat{x}}^2 \leq V(x,\hat{x},\theta)$ for all $\theta \in \Real^{l_{\theta}}$, implying that inequality \eqref{in:defV} holds with $\alpha(r) = \frac{\lambda_{\min}(\widehat{M})}{\lambda_{\max}(C_1^TC_1)}r$ for any $r \in \Real_{\geq 0}$, which is a convex function. We proceed to prove inequality \eqref{ineq:defDiss}. By the definition of $V$, one has 
	\begin{align*}
	&\partial_x V = 2(x-P\hat{x})^T\widehat{M}, 
	\partial_{\hat{x}}V = -2(x-P\hat{x})^T \widehat{M}P, \partial_{x,x}V = 2\widehat{M},
	\partial_{\hat{x},\hat{x}}V = 2P^T\widehat{M}P.\nonumber
	\end{align*}	
	Following the definition of $\infgen$, for any $x \in \Real^n,\hat{x} \in \Real^{\hat{n}}, \theta \in \Real^{l_{\theta}}$, one obtains:
		\begin{align*}
	\infgen{V}(x,\hat{x},\theta) &= 2(x-P\hat{x})^T\widehat{M}(Ax + E\varphi(Fx) + Bu + Dw)  - 2(x-P\hat{x})^T\widehat{M}P(\hat{A}\hat{x} + \hat{E}\varphi(\hat{F}\hat{x}) + \hat{B}\hat{u} +\hat{D}\hat{w}) \nonumber \\ &\quad+ G^T\widehat M G + \hat G^T P^T \widehat M P \hat G + 2(x-P\hat x)^T \widehat{M} \sum_{i=1}^{\mathsf r} \lambda_iR_i + \sum_{i=1}^{\mathsf r}\lambda_i R_i^T\widehat{M} R_i -
	2(x-P\hat x)^T \widehat{M} \sum_{i=1}^{\hat{\mathsf r}}\hat{\lambda}_iP\hat{R} \nonumber\\
	&\quad+ \sum_{i=1}^{\hat{\mathsf r}}\hat{\lambda}_i\hat{R}_i^T P^T\widehat{M} P \hat{R}_i + 2\theta^T\Lambda\Big(\Adyn\theta + \begin{bmatrix}B_1 & B_2\end{bmatrix}\begin{bmatrix} Ww - \hat{W}\hat{w} \\ C_2x - H\hat{C}_2\hat{x}\end{bmatrix}\Big).
	\end{align*}
	Given any $x \in \Real^n, \hat{x} \in \Real^{\hat{n}},$ and $\hat{u} \in \Real^{\hat{m}}$, we use the following {\it interface} function to choose $u \in \Real^m$:
	\small
	\begin{align}\label{eq:interface}
	u = K(x-P\hat{x}) + Q\hat{x} + \tilde{R}\hat{u} + L_1\varphi(t,Fx) - L_2\varphi(t,\hat{F}\hat{x}),~~
	\end{align}
	\normalsize
	where 
	$L_2$, $Q$, and $\tilde{R}$ are matrices of appropriate dimension.
	 Using the interface function in \eqref{eq:interface}, and the conditions \eqref{eq:dzw}, \eqref{eq:appabq}, \eqref{fpHatf}, \eqref{epHatEbl1l2}, and \eqref{pHatdzHatw}, one obtains:

	\begin{align}
 \infgen V(x,\hat{x},\theta) &=  2(x-P\hat{x})^T\widehat{M}\Big(A(x-P\hat{x}) + BK(x-P\hat{x}) 
	+ ZWw - Z\hat{W}\hat{w} + (B\tilde{R} - P\hat{B})\hat{u}
	+ (BL_1+E)\delta\varphi\Big) 
	\nonumber\\  &\quad   + G^T\widehat M G + \hat{G}^TP^T\widehat M P\hat G + \sum_{i=1}^{\mathsf r}\lambda_i R_i^T\widehat{M} R_i +
	\sum_{i=1}^{\hat{\mathsf r}}\hat{\lambda}_i\hat{R}_i^TP\widehat{M}P\hat{R}_i  +2(x-P\hat{x})^T\widehat{M}(\sum\limits_{i=1}^{\mathsf{r}}\lambda_i R_i\nonumber \\ 
	&\quad-\sum\limits_{i=1}^{\hat{\mathsf{r}}}\hat{\lambda}_iP\hat{R}_i ) + 2\theta^T\Lambda\Adyn\theta + 2\theta^T\Lambda B_1(Ww - \hat{W}\hat{w})\nonumber 
	 + 2\theta^T\Lambda B_2(C_2x - H\hat{C}_2\hat{x}),
	\end{align}
	\normalsize
	where $\delta\varphi = \varphi(t,Fx) - \varphi(t, \hat{F}\hat{x})$.
	 Using Young's inequality, Cauchy-Schwarz inequality, \eqref{ineq:bmi}, and \eqref{x12c2px12hHatc2}, one obtains the upper bound for $\infgen{V}(x,\hat{x},\theta)$ as given in \eqref{ineq:upperbound},
	\begin{figure*}[!t]
	\footnotesize
	\begin{align}
	\infgen{V}(x,\hat{x},\theta) &= 
	\begin{bmatrix}
		x - P\hat{x} \\ Ww - \hat{W}\hat{w} \\ \delta \varphi \\ \theta 
		\end{bmatrix}^T
		\begin{bmatrix}
		\Delta& \widehat{M}Z & \widehat{M}(BL_1+E)  &  C_2^TB_2^T\Lambda \\  
		Z^T\widehat{M} & 0 & 0 & B_1^T\Lambda \\ 
		(BL_1+E)^T\widehat{M} & 0 & 0 & 0 \\
		\Lambda B_2 C_2 & \Lambda B_1 & 0 & {\Adyn}^T\Lambda + \Lambda{\Adyn}
		\end{bmatrix}
	\begin{bmatrix} 
		x - P\hat{x} \\ Ww - \hat{W}\hat{w} \\  \delta \varphi \\ \theta 
		\end{bmatrix}
		\\ & \quad  + 2(x-P\hat{x})^T\widehat{M}(B\tilde{R} - P\hat{B})\hat{u} 
+2(x-P\hat{x})^T\widehat{M}\left(\sum\limits_{i=1}^{\mathsf{r}}\lambda_i R_i-\sum\limits_{i=1}^{\hat{\mathsf{r}}}\hat{\lambda}_iP\hat{R}_i\right)
	+ \mathsf{\tilde{c}}\nonumber	\\ 
	& \leq   
	\begin{bmatrix}
		x - P\hat{x} \\ Ww - \hat{W}\hat{w} \\ \delta \varphi \\ \theta 
	\end{bmatrix}^T
	 \begin{bmatrix}
	 -\hat{\kappa}\widehat{M} +p_{1} -  F^T M_{11} F & p_2 & -F^T M_{12} & C_2^TD_2^TXC_{\theta} \\
	 p_3  &  p_4 & 0 & D_1^TXC_{\theta} \\
	 -M_{12}^T F  & 0 & -M_{22} & 0\\
	 C_{\theta}^TXD_2C_2 & C_{\theta}^TXD_1 & 0 & {\Cdyn}^TX{\Cdyn} - \bar{\kappa}\Lambda
	 \end{bmatrix}
	\begin{bmatrix}
		x - P\hat{x} \\ Ww - \hat{W}\hat{w} \\ \delta \varphi \\ \theta 
	\end{bmatrix}  \nonumber \\
	&\quad + 2(x-P\hat{x})^T\widehat{M}(B\tilde{R} - P\hat{B})\hat{u} 
	+2(x-P\hat{x})^T\widehat{M}\left(\sum\limits_{i=1}^{\mathsf{r}}\lambda_i R_i-\sum\limits_{i=1}^{\hat{\mathsf{r}}}\hat{\lambda}_iP\hat{R}_i\right)
	+ \mathsf{\tilde{c}} \nonumber \\
	&\leq -(\hat{\kappa}-\pi -\pi^{\prime})(x-P\hat{x})^T\widehat{M}(x-P\hat{x}) + \frac{\norm{\sqrt{\widehat{M}} (B\tilde{R} - P\hat{B})}^2}{\pi}\norm{\hat{u}}^2  -\bar{\kappa}\theta^T\Lambda\theta\\	&\quad - 2\begin{bmatrix}
	 x-P\hat{x} \\ 
	 \delta\varphi
	\end{bmatrix}^T 
	\begin{bmatrix}
	F & 0_{l_k} \\
	0_{l_k} & I_{l_k}
	\end{bmatrix}^T \tilde M
	\begin{bmatrix}
	F& 0_{l_k} \\
	0_{l_k} & I_{l_k}
	\end{bmatrix}
	\begin{bmatrix}
	x-P\hat{x} \\ 
	\delta\varphi
	\end{bmatrix} \nonumber \\
	& \quad + \Bigg(\Cdyn\theta + \begin{bmatrix}
	D_1 & D_2
	\end{bmatrix}^T \begin{bmatrix}
	Ww - \hat{W}\hat{w}      \\
	C_2x - H\hat{C}_2\hat{x}
	\end{bmatrix}\Bigg)^T 
	X\Bigg(\Cdyn\theta + \begin{bmatrix}
	D_1 & D_2
	\end{bmatrix} \begin{bmatrix}
	Ww - \hat{W}\hat{w}      \\
	C_2x - H\hat{C}_2\hat{x}
	\end{bmatrix}\Bigg)  + \mathsf{\tilde{c}} +\frac{\norm{\small\sqrt{\widehat{M}}\left(\sum\limits_{i=1}^{\mathsf{r}}\lambda_i R_i-\sum\limits_{i=1}^{\hat{\mathsf{r}}}\hat{\lambda}_iP\hat{R}_i\right)}^2}{\pi^{\prime}} \nonumber \\
	&\leq -(\hat{\kappa}-\pi - \pi^{\prime})(x-P\hat{x})^T\widehat{M}(x-P\hat{x})  -\bar{\kappa}\theta^T\Lambda\theta+ \frac{\norm{\sqrt{\widehat{M}} (B\tilde{R} - P\hat{B})}^2}{\pi}\norm{\hat{u}}^2
	+ z^TXz  + \mathsf{\tilde{c}} + \mathsf{c^{\prime}} \nonumber \\\label{ineq:upperbound}
	&\leq -\tilde{\kappa}V(x,\hat{x},\theta) + \frac{\norm{\sqrt{\widehat{M}} (B\tilde{R} - P\hat{B})}^2}{\pi}\norm{\hat{u}}^2 +z^TXz + \mathsf{\tilde{c}} + \mathsf{c^{\prime}}
	\end{align}	
\horizontalEqSep
\end{figure*}
	where $\pi,\pi' \in \Real_{>0}$ satisfy $\pi + \pi^{\prime} < \hat{\kappa}$, $\tilde{\kappa} = \min\{\hat{\kappa} - \pi -\pi', \bar{\kappa}\}$, and 
	\small
	\begin{align}
	\mathsf{\tilde{c}} \!=&G^T\widehat M G \!+\! \hat{G}^TP^T\widehat M P\hat G \!+\! \sum_{i=1}^{\mathsf r}\lambda_i \!R_i^T\widehat{M} R_i\!+\!
	\sum_{i=1}^{\hat{\mathsf r}}\hat{\lambda}_i\!\hat{R}_i^TP^T\widehat{M}P\hat{R}_i \label{eq:c}, \\
	\mathsf{c^{\prime}} =& \frac{\norm{\small\sqrt{\widehat{M}}\left(\sum\limits_{i=1}^{\mathsf{r}}\lambda_i R_i-\sum\limits_{i=1}^{\hat{\mathsf{r}}}\hat{\lambda}_iP\hat{R}_i\right)}^2}{\pi^{\prime}}.
	\end{align}
	\normalsize
	Here, we have used the fact that for any $x \in \Real^n$ and any $\hat x \in \Real^{\hat n}$, one has \cite{accikmecse2011observers}, 	
	\begin{align}
	\begin{bmatrix}
	x-P\hat{x} \\ 
	\delta\varphi
	\end{bmatrix}^T 
	\begin{bmatrix}
	F  & 0_{l_k} \\
	0_{l_k} & I_{l_k}
	\end{bmatrix}^T \tilde M
		\begin{bmatrix}
	F  & 0_{l_k} \\
	0_{l_k} & I_{l_k}
	\end{bmatrix}
	\begin{bmatrix}
	x-P\hat{x} \\ 
	\delta\varphi
	\end{bmatrix} \geq 0.
	\end{align}
	Using the upper bound \eqref{ineq:upperbound}, the inequality \eqref{ineq:defDiss} is satisfied, implying that $V$ is an $\sstft$ from $\hat{\Sigma}$ to $\Sigma$, with respect to  $\Sigma_{\theta} = (A_{\theta}, B_{\theta}, C_{\theta}, D_{\theta})$, with the convex function $\eta(s) = \tilde{\kappa}s,$ concave function $\psi_{\mathsf{\mathsf{ext}}}(s) = \frac{\norm{\sqrt{\widehat{M}}(B\tilde{R}-P\hat{B})}^2}{\pi}s, \forall s \in \Real_{\geq 0}$, matrix $X$, and $\mathsf{c} = \mathsf{\tilde{c}} + \mathsf{c^{\prime}}$. 
\end{proof}
\begin{remark}
	Note that matrix $\tilde{R}$ is a free design parameter in the interface function. As explained in \cite{7857702} and \cite{girard2009hierarchical}, one can choose $\tilde{R}$ to minimize the function $\psi_{\mathsf{ext}}$ for $V$ and, hence, lower the upper bound on the error between the output behaviors of $\Sigma$ and $\hat{\Sigma}$. The choice of $\tilde{R}$ minimizing $\psi_{\mathsf{ext}}$ is given by 
	\begin{align}
		\tilde{R} = (B^T\widehat{M}B)^{-1}B^T\widehat{M}P\hat{B}. \label{eq:mintildeR}
	\end{align}
\end{remark}

\begin{remark}
	The constant $\mathsf{c}$, can be also minimized, thereby lowering the upper bound on the error between the output behaviours of $\Sigma$ and $\hat{\Sigma}$. One can choose $\hat{G}$ to be the zero matrix and choose $\hat{\lambda}$ and $\hat{\mathsf{R}}$ to solve the following optimization problem:
	\small
	\begin{align}
	&\argmin\limits_{\hat{\mathsf{R}}, \hat{\lambda} > 0} \quad \sum\limits^{\hat{\mathsf{r}}}_{i=1}\hat{\lambda}_i\hat{R}_i^TP^T\widehat{M}P\hat{R}_i - \frac{2(\sum\limits^{\mathsf{r}}_{i=0} \lambda_iR_i^T)\widehat{M}P(\sum\limits^{\hat{\mathsf{r}}}_{i=0}\hat{\lambda}_i\hat{R}_i)}{\pi^{\prime}}   +
	\frac{(\sum\limits^{\hat{\mathsf{r}}}_{i=1}\hat{\lambda}_i\hat{R}_i^T)P^T\widehat{M}P(\sum\limits^{\hat{\mathsf{r}}}_{i=1}\hat{\lambda}_i\hat{R}_i)}{\pi^{\prime}},
	\label{opt:rhatlambda}
	\end{align}
	\normalsize
	where $\hat{\lambda} = \{\hat{\lambda}_1, \dots, \hat{\lambda}_{\hat{\mathsf r}}\}$ and $\hat{\mathsf{R}} = \{\hat{R}_1,\dots,\hat{R}_{\hat{\mathsf r}}\}$. This optimization problem is, in general, a non-convex one. 
\end{remark}

In the following theorem we show that conditions \eqref{eq:appabq}, \eqref{eq:c1pHatc1}, \eqref{x12c2px12hHatc2},  \eqref{fpHatf}, and \eqref{epHatEbl1l2} are not only sufficient, but also necessary for \eqref{eq:quadraticstoragefunction} to be an $\sstft$ from $\hat{\Sigma}$ to $\Sigma$, provided that the interface function is as in $\eqref{eq:interface}$ for some matrices $K, Q, \tilde{R}, L_1,$ and  $L_2,$ of appropriate dimensions. 
\begin{theorem}
	Let $\Sigma = (A,B,C_1,C_2,D,E,F,G,\mathsf{R},\varphi, \lambda)$ and $\hat{\Sigma} = (\hat{A},\hat{B},\hat{C}_1,\hat{C}_2,\hat{D},\hat{E},\hat{F},\hat{G},\hat{\mathsf{R}},\varphi,  \hat{\lambda})$ with the same external output space dimension. Assume that $G = \hat{G} = 0$, and $R_i = \hat{R}_i = 0$ $\forall i\in[1;\hat{\mathsf r}]$, where $0$ represents the zero matrices of appropriate dimensions.  Suppose that $V$, defined in \eqref{eq:quadraticstoragefunction}, is an $\sstft$ from $\hat{\Sigma}$ to $\Sigma$, with respect to $\Sigma_{\theta} = (\Adyn, \Bdyn, \Cdyn, \Ddyn)$, with the interface function given in \eqref{eq:interface}. Then equations \eqref{eq:appabq}, \eqref{eq:c1pHatc1}, \eqref{x12c2px12hHatc2},  \eqref{fpHatf}, and \eqref{epHatEbl1l2} hold.
\end{theorem}
\begin{proof}
	Since $V$ is an $\sstft$ from $\hat{\Sigma}$ to $\Sigma$, there exists a \kinf \ function $\alpha$ such that $\norm{C_1x - \hat{C}_1\hat{x}}^2 \leq \alpha^{-1}(V(x,\hat{x},\theta))$ for any $x \in \Real^n$, any $\hat{x} \in \Real^{\hat{n}},$ and any  $\theta \in \Real^{l_{\theta}}$. From \eqref{eq:quadraticstoragefunction}, it follows that $\norm{C_1P\hat{x} - \hat{C}_1\hat{x}}^2 \leq \alpha^{-1}(V(P\hat{x}, \hat{x}, 0)) = 0$ holds for all $\hat{x} \in \Real^{\hat{n}}$ which implies \eqref{eq:c1pHatc1}. 
	Let us assume that $D_2^TXD_2 \neq 0$. To prove \eqref{x12c2px12hHatc2}, we consider the inputs  $w \equiv 0, \hat{w} \equiv 0, \hat{u} \equiv 0$, and choose $x = P\hat{x}$ and $\theta = 0$ in \eqref{ineq:defDiss}. One has:
	\begin{align}
	0 &\leq (C_2P\hat{x} - H\hat{C}_2\hat{x})^TD_2^TXD_2(C_2P\hat{x} - H\hat{C}_2\hat{x}),
	\end{align}
	for all $\hat{x} \in \Real^{\hat{n}}.$ Since $D_2^TXD_2 \preceq 0$, and  $D_2^TXD_2 \neq 0$ by assumption, one obtains $C_2P-  H\hat{C}_2 = 0$, which implies \eqref{x12c2px12hHatc2}.	Consider the input signals $\hat{\upsilon} \equiv 0, \omega \equiv 0, \hat{\omega} \equiv 0$. It can be easily seen that the subspace $\{(x,\hat{x}, \theta)    :  x = P\hat{x}, \theta = 0\}\subseteq \Real^{n} \times \Real^{\hat{n}} \times \Real^{l_{\theta}} $ is invariant \cite{khalil1996nonlinear}, which implies that when $\xi(0) = P\hat{\xi}(0)$ and $\xi_{\theta}(0) = 0$, one has:	
	\begin{align}
	\xi(t) = P\hat{\xi}(t), \quad \xi_{\theta}(t) = 0,\quad \diff\xi(t) = P\diff\hat{\xi}(t),
	\end{align}
	for all $t \in \Real_{\geq 0}$, from which we derive that
	\begin{align}
	&(A P\hat{\xi}(t) + BQ\hat{\xi}(t) + BL_1\varphi(t,F\xi(t)) - BL_2\varphi(t,\hat{F}\hat{\xi}(t)) + E\varphi(t,FP\hat{\xi}(t)))\diff t = (P\hat{A}\hat{\xi}(t) + P\hat{E}\varphi(t,\hat{F}\hat{\xi}(t)))\diff t,
	\end{align}
	for all $t \in \Real_{\geq 0}$, thus implying \eqref{eq:appabq}, \eqref{fpHatf}, and \eqref{epHatEbl1l2}.
\end{proof}

\subsection{Geometric interpretation of different conditions}

In this section, we provide geometric conditions on matrices appearing on the definition of $\hat \Sigma$, of stochastic storage function and its corresponding interface function. The geometric conditions facilitate the construction of the abstraction.
First, we recall the following result from \cite{girard2009hierarchical}, providing necessary and sufficient conditions for the existence of $\hat A$ and $Q$ satisfying \eqref{eq:appabq}.
\begin{lemma}
	Consider matrices $A$, $B$, and $P$. There exist matrices $\hat A$ and $Q$ satisfying \eqref{eq:appabq} if and only if
	\begin{align}
		\mathrm{im} \ AP \subseteq \mathrm{im} \ P + \mathrm{im} \ B \label{set:imapimpimb}.
	\end{align}
	\label{lem:appb}
\end{lemma}
Similarly, we provide necessary and sufficient conditions for the existence of $\hat{C}_2$ and $\hat{E}$, $L_2$ satisfying \eqref{x12c2px12hHatc2} and \eqref{epHatEbl1l2}, respectively.

\begin{lemma}
	Given $P$ and $C_2$, there exists matrix $\hat{C}_2$ satisfying \eqref{x12c2px12hHatc2} if and only if
	\begin{align}
		\mathrm{im} \ C_2 P \subseteq \mathrm{im} \ H \label{set:imc2pimh}
	\end{align}
	for some matrix $H$ of appropriate dimension.
	\label{lem:c2ph}
\end{lemma}
\begin{lemma}
	Given $P$, $B$, and $L_1$, there exist matrices $\hat{E}$ and $L_2$  satisfying \eqref{epHatEbl1l2} if and only if
	\begin{align}
		\mathrm{im} \  E \subseteq \mathrm{im} \ B + \mathrm{im} \ P.
		\label{set:imbbl3p}
		\end{align}
		\label{lem:ebbl3p}
\end{lemma}
Lemmas \ref{lem:appb}, \ref{lem:c2ph}, and \ref{lem:ebbl3p} provide sufficient and necessary conditions on $P$ and $H$, resulting in the construction of matrices $\hat{A}$, $\hat{C}_2$, and $\hat{E}$ and matrices $Q$ and $L_2$ appearing in the interface function \eqref{eq:interface}. 
The next lemma provides a sufficient and necessary condition on the existence of $\hat{D}$ satisfying \eqref{pHatdzHatw}.
\begin{lemma}
	Given $Z$, there exists matrix $\hat{D}$ satisfying \eqref{pHatdzHatw} if and only if
	\begin{align}
	\mathrm{im} \ Z\hat{W} \subseteq \mathrm{im} \ P,
	\label{set:imzwimp} 
	\end{align}
	for some matrix $\hat{W}$ of appropriate dimension.
\end{lemma}
Although condition \eqref{set:imzwimp} is readily satisfied by choosing $\hat W=0$, one should preferably aim at finding a nonzero $\hat W$ with the highest possible rank
to facilitate later the satisfaction of compositionality condition \eqref{eq:dycinterconnected2}.

%
We summarize the construction of abstraction $\hat{\Sigma}$, stochastic storage function $V$ in \eqref{eq:quadraticstoragefunction}, and its corresponding interface function in \eqref{eq:interface} in Table \ref{table:construction}.
\begin{table}
\caption{\small Construction of $\hat{\Sigma} = (\hat{A},\hat{B}, \hat{C}_1, \hat{C}_2, \hat{D}, \hat{E}, \hat{F}, \hat{G},\hat{\mathsf{R}},\varphi, \hat{\lambda})$ together with the corresponding stochastic storage function $V$ in \eqref{eq:quadraticstoragefunction}, with $\Sigma_{\theta} = (A_{\theta}, B_{\theta}, C_{\theta}, D_{\theta})$, and interface function in \eqref{eq:interface} for a given $\Sigma = (A,B,C_1, C_2, D, E, F, G, \mathsf R, \varphi, \lambda)$. \normalsize}\label{table:construction}
\centering
	\begin{tabular}{ |p{8.3cm}| } 
		\hline
		\begin{enumerate}
			\item Compute matrices $\widehat{M}, K, L_1, X, A_{\theta}, C_{\theta},$ $B_\theta:=[B_1~B_2], D_\theta=[D_1~D_2],  \Lambda$, and $Z$ satisfying \eqref{eq:dzw} and \eqref{ineq:bmi};
			\item Pick an injective $P$ with the lowest rank satisfying \eqref{set:imapimpimb}, \eqref{set:imc2pimh}, \eqref{set:imbbl3p}, and \eqref{set:imzwimp} ;
			\item Choose $\hat{A}$ and $Q$ according to  \eqref{eq:appabq};
			\item Choose $L_2$ and $\hat{E}$ according to \eqref{epHatEbl1l2};
			\item Compute $\hat{F} = FP$;
			\item Compute $\hat{C}_1 = C_1P$;
			\item Choose $\hat{G} = 0$. Choose $\hat{\mathsf{R}} = \{\hat{R}_1,\dots,\hat{R}_{\hat{\mathsf{r}}}\}$ and $\hat{\lambda} = \{\hat{\lambda}_1,\dots,\hat{\lambda}_{\hat{\mathsf{r}}} \}$ according to \eqref{opt:rhatlambda}; 
			\item Choose $\hat{C}_2$ satisfying $H\hat{C}_2 = C_2P$ for some $H$;
			\item Choose $\hat{D}$ satisfying $P\hat{D} = Z\hat{W}$ for some $\hat{W}$ with the highest rank;
			\item Choose $\hat{B}$ freely (e.g. $\hat{B}=I_{\hat n}$ making $\hat \Sigma$ fully actuated);
			\item Compute $\tilde{R}$, appearing in \eqref{eq:interface}, according to \eqref{eq:mintildeR};
		\end{enumerate} \\
		\hline
	\end{tabular}
\end{table}
\section{Examples}
\subsection{Example 1} 
Consider the following system: 
\begin{IEEEeqnarray*}{c}
	\Sigma:\left\{
	\begin{IEEEeqnarraybox}[\relax][c]{rCl}
		\diff \xi(t) &=& (-L\xi(t) + \upsilon(t) + \varPhi(\xi(t)))\diff t + G\diff W_t + R\diff P_t \nonumber,\\
		\zeta(t) &=& C\xi(t),%
	\end{IEEEeqnarraybox}\right.
\end{IEEEeqnarray*}
where the matrices $C \in \Real^{q\times n}$, $L \in \Real^{n\times n}$, $G \in \Real^{n \times 1}$, $R \in \Real^{n \times 1}$, and  the vector valued function $\varPhi : \Real^n \rightarrow \Real^n$ are defined as follows:
\begin{align}
L &= \begin{bmatrix}
n - 1 & -1 & \dots & \dots & -1 \\
-1 & n-1 & -1 & \dots & -1 \\
-1 & -1 & n-1 & \dots & -1 \\
\vdots &  & \ddots & \ddots & \vdots \\
-1 & \dots & \dots & -1 & n-1
\end{bmatrix},  \nonumber \\ C &= \mathsf{diag}(C_{11},\ldots,C_{1N}), 
G = \varpi\vec{1}_n, R = \tau\vec{1}_n, \nonumber \\
\varPhi(\xi) &= [\vec{1}_{n_1}\sin(\vec{1}_{n_1}^T\xi_1);\dots;\vec{1}_{n_N}\sin(\vec{1}^T_{n_N}\xi_N)],
\end{align}
where $\tau$, $\varpi$ $\in \Real_{>0}$, $C_{1i} \in \Real^{q_{1i}\times n_i}$, and $\xi_1,\dots,\xi_N$ are defined as follows: $\xi$ is partitioned as $\xi = [\xi_1;\dots;\xi_N]$ and $v$ as $\upsilon = [\upsilon_1;\dots;\upsilon_N]$, where $\xi_i$ and $\upsilon_i$ are both taking values in $\Real^{n_i}, \forall i \in [1;N]$. Assume the rate of the Poisson process $P_t$ is $\lambda$. 
By introducing $\Sigma_i = (0_{n_i}, I_{n_i}, C_{1i}, I_{n_i}, I_{n_i},  \vec{1}_{n_i}, \vec{1}_{n_i},\varpi \vec{1}_{n_i},\tau \vec{1}_{n_i}, \varPhi_i,\lambda)$ satisfying
\begin{IEEEeqnarray*}{c}
	\Sigma_i:\left\{
	\begin{IEEEeqnarraybox}[\relax][c]{rCl}
		\diff \xi_i(t) &=& (\omega_i(t) + \upsilon_i(t) + \vec{1}_{n_i}\varPhi_i(\vec{1}_{n_i}^T\xi_i))\diff t \nonumber \\&&+ \varpi \vec{1}_{n_i} \diff W_t + \tau\vec{1}_{n_i} \diff P_t\nonumber,\\
	\zeta_{1i}(t)&=& C_{1i}\xi_i(t), \\%
	\zeta_{2i}(t) &=& \xi_i(t),
	\end{IEEEeqnarraybox}\right.
\end{IEEEeqnarray*}
where $\varPhi_i: \Real \rightarrow \Real$ is defined by  $\varPhi_i(x) = \sin(x)$, $\forall i \in [1;N],$ one can verify that $\Sigma = \mathcal{I}(\Sigma_1,\dots,\Sigma_N)$ where the coupling matrix $M$ is given by $M = -L$. 
We consider a deterministic scalar abstraction $\hat{\Sigma}_i = (0,1,C_{1i}\vec{1}_{n_i}, 1,1,1,1,0,0,\varPhi_i,0)$ satisfying
\begin{align}
\hat{\Sigma}_i : 
\begin{cases}
\diff \hat{\xi}_i(t) = (\hat{\omega}_i(t)+ \hat{\upsilon}_i(t) + \varPhi_i(\hat{\xi}_i))\diff t \\
\hat{\zeta}_{1i}(t) = C_{1i}\vec{1}_{n_i}\hat{\xi}_i(t) \nonumber \\
\hat{\zeta}_{2i}(t) = \hat{\xi}_i(t),
\end{cases}
\end{align}
for any $i \in [1;N]$. The function  $V(x_i, \hat x_i, \theta_i) = (x-\vec{1}_{n_i}\hat{x}_i)^T(x-\vec{1}_{n_i}\hat{x}_i)$ (i.e. $\widehat{M}_i = I_{n_i}$, $P_i = \vec{1}_{n_i}, \Lambda = 0$) is an $\sstft$ from $\hat{\Sigma}_i$ to $\Sigma_i$, with respect to $\Sigma_{\theta_i} = (A_{\theta_i}, B_{\theta_i}, C_{\theta_i}, D_{\theta_i})$, $\forall i \in [1;N]$, with the following parameters 
\begin{align}
&K_i = -\chi I_{n_i}, \hat{\kappa}_i = 2\chi - 2\lambda\tau - \varpi^2 - \lambda\tau^2, Z_i = I_{n_i}, \nonumber \\& W_i = I_{n_i}, Q_i = 0_{n_i}, H_i = \hat{W}_i = \vec{1}_{n_i}, L_{1i} = -\vec{1}_{n_i}, A_{\theta_i} = 0,\nonumber \\&B_{\theta_i} = 0, C_{\theta_i} = 0, D_{\theta_i} = I_{2n_i}, \bar{\kappa} = 0, X_i = \begin{bmatrix}
	  0_{n_i} & I_{n_i} \\ I_{n_i} & 0_{n_i}
	\end{bmatrix}
\end{align}
for some $\chi > \lambda\tau + \frac{\varpi^2}{2} + \frac{\lambda\tau^2}{2}$, and with $\alpha_i(r) = \frac{1}{\lambda_{\max}(C_{1i}^TC_{1i})}r$, $\eta_i(r) = (2\chi - 2\lambda\tau -\varpi^2 - \lambda\tau^2)r, \psi_{i\mathsf{ext}}(r) = 0$, $\forall r \in \Real_{\geq 0},$ and $\mathsf{c}_i = \tau^2 + \varpi^2$. Inputs $u_i \in \Real^{n_i}$ is given via the interface function in (\ref{eq:interface}) as (i.e. $\tilde{R}_i = \vec{1}_{n_i}, L_{2i} = \vec{1}_{n_i})$
\begin{align}
u_i &= -\chi(x_i - \vec{1}_{n_i}\hat{x}_i) + \vec{1}_{n_i}\hat{u}_i - \vec{1}_{n_i}\varPhi_i(\vec{1}_{n_i}^Tx_i) + \vec{1}_{n_i}\varPhi_i(\hat{x}_i).
\label{eq:int1}
\end{align}
\normalsize
By selecting $\mu_1 = \ldots = \mu_N = 1$, the function $V(x,\hat{x},\theta) = \sum^N_{i=1}\mu_i V_i(x_i,\hat{x}_i,\theta_i)$ is an $\ssft$ function from $\hat{\Sigma}$ to $\Sigma$, where $\hat{\Sigma}$ is the interconnection of the abstract subsystems $\hat{\Sigma} = \mathcal{I}(\hat{\Sigma}_1,\dots,\hat{\Sigma}_N)$ with a coupling matrix $\hat{M}$, satisfying condition \eqref{eq:dycinterconnected2} as the following
\begin{align}
\label{eq:example_Mhat}
-L\mathsf{diag}(\vec{1}_{n_1},\ldots,\vec{1}_{n_N})
=
\mathsf{diag}(\vec{1}_{n_1},\ldots,\vec{1}_{n_N}) \hat{M} .
\end{align}
A matrix $\hat{M}$ exists satisfying \eqref{eq:example_Mhat} if there exist $N$ {\it equitable partitions} of the graph described by the Laplacian matrix $L$, which is always true here because $L$ represents a fully connected graphs, as explained in \cite{godsil2013algebraic}.
It can be easily seen that condition \eqref{eq:dyninterconnected1} reduces to 
\begin{align}
\begin{bmatrix} 
-L \\
I_n 
\end{bmatrix}^T 
\begin{bmatrix}
0 & I_n \\
I_n & 0
\end{bmatrix} \begin{bmatrix} 
-L \\
I_n 
\end{bmatrix} = -(L + L^T) \preceq 0,
\end{align}
which always holds since $L = L^T \succeq 0$, which is always true for Laplacian matrices of undirected graphs \cite{godsil2013algebraic}.
	\begin{figure}[t]
			\begin{subfigure}[t]{0.5\textwidth}
			\centering
			\includegraphics[scale=0.3]{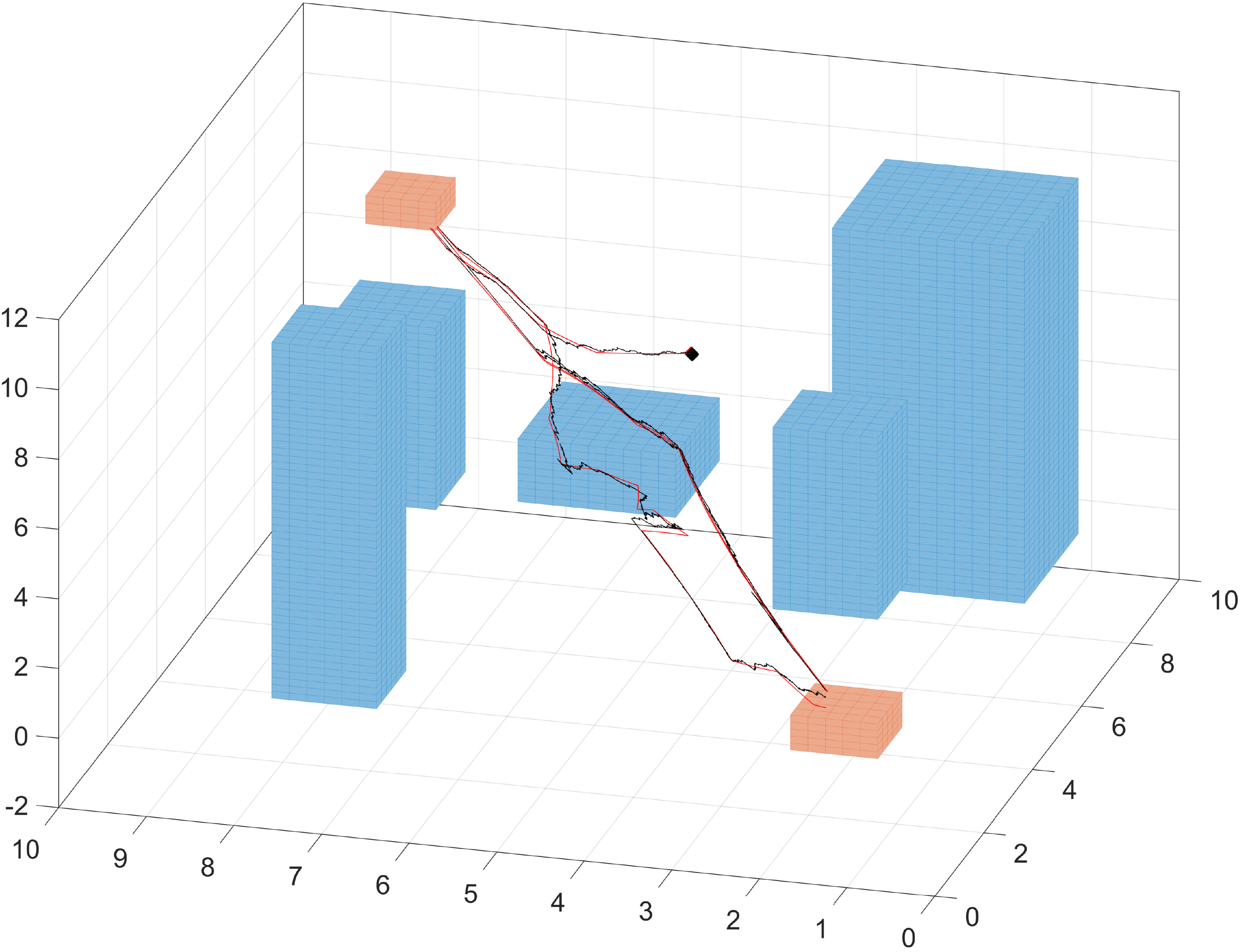}
		\end{subfigure}

		\caption{\small The figure shows the output trajectories of the abstract (red) and one realization of the concrete (black) interconnected systems. The initial point of the trajectories is represented by the diamond.
			\normalsize}
		\label{fig:fig1}
	\end{figure}
\subsubsection{Controller synthesis}
Now, we synthesize a controller for the abstract interconnected system $\hat{\Sigma} = \mathcal{I}(\hat{\Sigma}_1,\dots,\hat{\Sigma}_N)$ to enforce a specification, and then refine the designed controller to the one for the concrete interconnected system. We fix $n$ = 9, $N$ = 3, $\tau = 0.2, \varpi = 0.4, \lambda = 1, \chi = 10$ and 
\begin{align}
C = \begin{bmatrix}
1 & 0 & 0 & 0 & 0 & 0 & 0 & 0 & 0\\
0 & 0 & 0 & 0 & 1 & 0 & 0 & 0 & 0\\
0 & 0 & 0 & 0 & 0 & 0 & 0 & 0 & 1
\end{bmatrix},
\end{align}
$$C_{11} = \begin{bmatrix}1 & 0 & 0 \end{bmatrix}, C_{12} =\begin{bmatrix}0 & 1 & 0 \end{bmatrix}, C_{13} = \begin{bmatrix}0 & 0 & 1 \end{bmatrix}.$$
 We synthesize a controller using toolbox \texttt{SCOTS} \cite{rungger2016scots} to enforce the following linear temporal logic specification \cite{katoen08} over the outputs of $\hat {\Sigma}$:
\begin{align}
\Psi = \square S \wedge \left(\bigwedge\limits_{i = 1}^5 \square (\lnot O_i)\right) \wedge \square \lozenge T_1 \wedge \square \lozenge T_2,
\end{align}

which can be interpreted as follows: the output trajectory of the closed loop system evolves inside the set $S$, avoids obstacles $O_i, i \in [1;5],$ indicated with blue boxes in Figure \ref{fig:fig1}, and visits $T_i, i \in [1;2]$ infinitely often, indicated with red boxes in Figure \ref{fig:fig1}.  We use \eqref{eq:int1}
to generate the corresponding input enforcing this specification over the original system $\Sigma$.
\subsection{Example 2}
In this part, we provide compositional abstractions of a network of subsystems wherein the joint dissipativity property of each concrete subsystem and its abstraction is only concluded with respect to a linear control system $\Sigma_{\theta}$ rather than a static map. 
Consider an interconnection of $N$ control subsystems $\Sigma_i$, where each $\Sigma_i$ is given by
\small
\begin{align}
\Sigma_i : \begin{cases}
\diff \xi_i(t) = (A_i\xi_i(t) + B_i\upsilon_i(t) + D_i\omega_i(t))\diff t, \\ 
\zeta_{1i}(t) = C_{1i}\xi_i(t), \\
\zeta_{2i}(t) = \xi_i(t),
\end{cases}
\end{align}
\normalsize
where 
\begin{align}
A_i = \begin{bmatrix}
0_{n_i} & I_{n_i}\\
-I_{n_i} & -0.5I_{n_i}
\end{bmatrix},
B_i = D_i = \begin{bmatrix}
0_{n_i} \\
I_{n_i}
\end{bmatrix},
C_{1i} = \begin{bmatrix}
0_{n_i} \\ e_{n_i}
\end{bmatrix}^T,
\end{align}
 vector $e_{n_i}$ represents a column vector whose first element is 1 and remaining elements are zero.  For the sake of simulation we choose $N$ = 3, $n_i = 10$, $\forall i \in [1;N]$. 
 We consider the following abstract system $\hat{\Sigma}_i$, 
\small
\begin{IEEEeqnarray}{c}
	\hat{\Sigma}_i: \left\{ \begin{IEEEeqnarraybox}[\relax][c]{rCl}
	\diff \hat{\xi}_i(t) &=& \left(\begin{bmatrix}
		0 & 1 \\
		-1 & -0.5
	\end{bmatrix}\hat{\xi}_i(t) + 
	\begin{bmatrix}
		0 \\
		1
	\end{bmatrix}\hat{\upsilon}_i(t) + 
	\begin{bmatrix}
		0 \\
		1
	\end{bmatrix}\hat{\omega}_i(t) \right) \diff t, \nonumber \\
	\hat{\zeta}_{1i}(t) &=& \begin{bmatrix} 0 & 1 \end{bmatrix}\hat{\xi}_i(t), \nonumber \\
	\hat{\zeta}_{2i}(t) &=& \hat{\xi}_i(t).
	\end{IEEEeqnarraybox}\right.
\end{IEEEeqnarray}
\normalsize
We restrict $K_i$ for each $i \in [1;N]$ appearing in \eqref{eq:interface} such that the last $n_i$ columns are identically zero. This restriction can appear in practice when for example only some state variables are available to be measured. With this restriction on the structure of $K_i$, one cannot find a storage function with $C_{\theta_i}  = 0$ in this example. It can be shown that the function $$V_i(x_i,\hat{x}_i,\theta_i) = (x_i-P\hat{x}_i)^T\widehat{M}(x_i-P\hat{x}_i) + \theta_i^T\Lambda\theta_i$$ is an $\sstft$ from $\hat{\Sigma}_i$ to $\Sigma_i$, with respect to $\Sigma_{\theta_i} = (A_{\theta_i}, B_{\theta_i}, C_{\theta_i}, D_{\theta_i})$, $\forall i \in [1;N]$, with the following parameters 
\small
\begin{align}
&\widehat{M}_i = \begin{bmatrix}
2I_{n_i} & I_{n_i} \\
I_{n_i} & I_{n_i}
\end{bmatrix}, 
P_i = \begin{bmatrix}
\vec{1}_{n_i} & \vec{0}	_{n_i} \\
\vec{0}_{n_i} & \vec{1}_{n_i}
\end{bmatrix}, K_i = \begin{bmatrix}
-0.5 I_{n_i} & 0_{n_i}
\end{bmatrix},\nonumber \\&\hat{\kappa}_i = 0.1, W_i = I_{n_i}, Q_i = 0,  
H_i = \hat{W}_i = \vec{1}_{n_i}, L_{1i} = 0,\Lambda = I_{2n_i}, \\&A_{\theta_i} = -5I_{2n_i}, B_{\theta_i} = \begin{bmatrix}
0_{n_i} & -4.14I_{n_i} \\
0_{n_i} & 11.51I_{n_i}
\end{bmatrix}, C_{\theta_i} = 0.1I_{2n_i}, \\&D_{\theta_i} =\begin{bmatrix}
0_{n_i} & I_{n_i} \\
0_{n_i} & I_{n_i}
\end{bmatrix},  X_i = \begin{bmatrix}
	9.47785I_{n_i} &  -7.4055I_{n_i} \\  -7.4055I_{n_i} & 1.6526I_{n_i}
	\end{bmatrix}
	,\bar{\kappa}_i = 1,
\label{ex3param}
\end{align}
\normalsize
with $\alpha_i(r) = \frac{\lambda_{\min} (\widehat{M}_i)}{\lambda_{\max}(C_{1i}^TC_{1i})}r$, $\eta_i(r) = 0.1r$, $\psi_{i\mathsf{ext}} = 0,  \forall r \in \Real_{\geq 0}$, and $c_i =0$.  
Functions $u_i \in \Real^{n_i}$ are given via the interface function:
\begin{align}
u_i = -K_i(x_i - P_i\hat{x}_i) + \vec{1}_{n_i}\hat{u}_i,
\end{align}
(i.e. $\tilde{R}_i = \vec{1}_{n_i}, L_{2i} = 0$). 
With the interconnection matrix $M$ given by \small \begin{align}
M &= \begin{bmatrix}
-2 & 1 & 0 & 0 & \dots & 1 \\
1 & -2 & 1 & 0 & \dots & 0 \\
0 & 1 & -2 & 1 & \dots & 0 \\
\vdots & & & \ddots & & \\ 
& & & & \ddots & \\
1 & 0 & 0 & \dots & 1 & -2 \\
\end{bmatrix}.
\end{align} 
\normalsize
and by selecting $\mu_1 = \dots = \mu_N = 1$, it can be verified that the function $V = \sum\limits_{i=1}^N\mu_iV_i(x_i,\hat{x}_i,\theta_i) + \theta^T\theta$, where $\theta = [\theta_i;\dots;\theta_N]$, is an $\ssft$ from $\hat{\Sigma}$ to $\Sigma$, where $\hat{\Sigma}$ is the interconnection of the abstract subsystems $\hat{\Sigma} = \mathcal{I}(\hat{\Sigma}_1,\dots,\hat{\Sigma}_N)$ with the coupling matrix $\hat{M}$ given by
\begin{align}
\hat{M} = \begin{bmatrix}
-2 & 1 & 1 \\
1 & -2 & 1\\
1 & 1 & -2
\end{bmatrix},
\end{align}
satisfying conditions \eqref{eq:dyninterconnected1} and \eqref{eq:dycinterconnected2}. In the simulation, the input signal to the abstract system is chosen arbitrarily as $\hat{\upsilon}(t) = [\sin(t);0.1\mathsf{e}^{-t};-t].$ Figure \ref{error} shows the evolution of the absolute value of the error between the output trajectories of the concrete interconnected system and its abstraction. One can readily verify that the error is always bounded by the computed error bound in Theorem \ref{theorem1}. 
	\begin{figure}[t!]
		\centering
		\centering
		\includegraphics[scale=0.4]{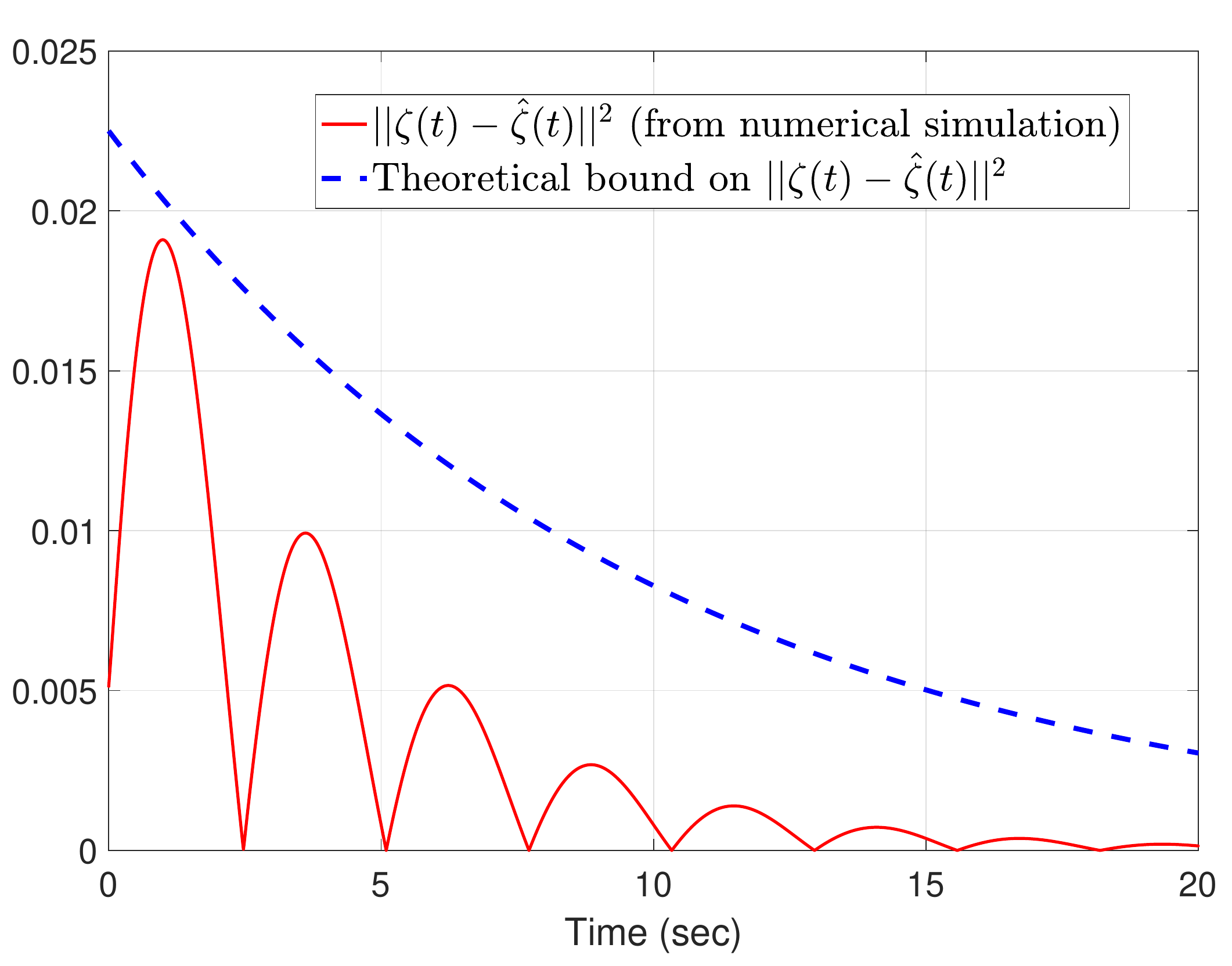}
	\caption{\small The evolution of ${\norm{\zeta(t) - \hat{\zeta}(t)}^2}$,  where $\zeta(t) = [\zeta_{11}(t);\dots;\zeta_{1N}(t)]$, and $\hat{\zeta}(t) = [\hat{\zeta}_{11}(t);\dots;\hat{\zeta}_{1N}(t)]$, and the theoretical upper bound obtained for this example according to \eqref{eq:bound_output}.
		 \normalsize}. \label{error}
\end{figure}
\bibliographystyle{IEEEtran}
\bibliography{bib}

\begin{thebibliography}{10}
\providecommand{\url}[1]{#1}
\csname url@samestyle\endcsname
\providecommand{\newblock}{\relax}
\providecommand{\bibinfo}[2]{#2}
\providecommand{\BIBentrySTDinterwordspacing}{\spaceskip=0pt\relax}
\providecommand{\BIBentryALTinterwordstretchfactor}{4}
\providecommand{\BIBentryALTinterwordspacing}{\spaceskip=\fontdimen2\font plus
\BIBentryALTinterwordstretchfactor\fontdimen3\font minus
  \fontdimen4\font\relax}
\providecommand{\BIBforeignlanguage}[2]{{%
\expandafter\ifx\csname l@#1\endcsname\relax
\typeout{** WARNING: IEEEtran.bst: No hyphenation pattern has been}%
\typeout{** loaded for the language `#1'. Using the pattern for}%
\typeout{** the default language instead.}%
\else
\language=\csname l@#1\endcsname
\fi
#2}}
\providecommand{\BIBdecl}{\relax}
\BIBdecl

\bibitem{pola2016symbolic}
G.~Pola, P.~Pepe, and M.~D. Di~Benedetto, ``Symbolic models for networks of
  control systems,'' \emph{IEEE Transactions on Automatic Control}, vol.~61,
  no.~11, pp. 3663--3668, November 2016.

\bibitem{tazaki2008bisimilar}
Y.~Tazaki and J.-i. Imura, ``Bisimilar finite abstractions of interconnected
  systems,'' in \emph{International Workshop on Hybrid Systems: Computation and
  Control}.\hskip 1em plus 0.5em minus 0.4em\relax Springer, 2008, pp.
  514--527.

\bibitem{7496809}
M.~Rungger and M.~Zamani, ``Compositional construction of approximate
  abstractions of interconnected control systems,'' \emph{IEEE Transactions on
  Control of Network Systems}, vol.~5, no.~1, pp. 116--127, March 2018.

\bibitem{zamani2015approximations}
M.~Zamani, M.~Rungger, and P.~M. Esfahani, ``Approximations of stochastic
  hybrid systems: A compositional approach,'' \emph{IEEE Transactions on
  Automatic Control}, vol.~62, no.~6, pp. 2838--2853, June 2017.

\bibitem{das2004some}
K.~C. Das and P.~Kumar, ``Some new bounds on the spectral radius of graphs,''
  \emph{Discrete Mathematics}, vol. 281, no.~1, pp. 149--161, 2004.

\bibitem{7857702}
M.~Zamani and M.~Arcak, ``Compositional abstraction for networks of control
  systems: A dissipativity approach,'' \emph{IEEE Transactions on Control of
  Network Systems}, vol.~PP, no.~99, pp. 1--1, 2017.

\bibitem{awanIFAC17}
A.~U. Awan and M.~Zamani, ``Compositional abstractions of networks of
  stochastic hybrid systems: A dissipativity approach,''
  \emph{IFAC-PapersOnLine}, vol.~50, no.~1, pp. 15\,804 -- 15\,809, 2017.

\bibitem{tippett2011dissipativity}
M.~J. Tippett and J.~Bao, ``Dissipativity based analysis using dynamic supply
  rates,'' \emph{IFAC Proceedings Volumes}, vol.~44, no.~1, pp. 1319--1325,
  2011.

\bibitem{oksendal2005applied}
B.~K. {\O}ksendal and A.~Sulem, \emph{Applied stochastic control of jump
  diffusions}.\hskip 1em plus 0.5em minus 0.4em\relax Springer, 2005, vol. 498.

\bibitem{arcak2016networks}
M.~Arcak, C.~Meissen, and A.~Packard, \emph{Networks of dissipative systems:
  compositional certification of stability, performance, and safety}.\hskip 1em
  plus 0.5em minus 0.4em\relax Springer, 2016.

\bibitem{boyd2004convex}
S.~Boyd and L.~Vandenberghe, \emph{Convex optimization}.\hskip 1em plus 0.5em
  minus 0.4em\relax Cambridge university press, 2004.

\bibitem{accikmecse2011observers}
B.~A{\c{c}}{\i}kme{\c{s}}e and M.~Corless, ``Observers for systems with
  nonlinearities satisfying incremental quadratic constraints,''
  \emph{Automatica}, vol.~47, no.~7, pp. 1339--1348, 2011.

\bibitem{girard2009hierarchical}
A.~Girard and G.~J. Pappas, ``Hierarchical control system design using
  approximate simulation,'' \emph{Automatica}, vol.~45, no.~2, pp. 566--571,
  2009.

\bibitem{khalil1996nonlinear}
H.~K. Khalil, ``Noninear systems,'' \emph{Prentice-Hall, New Jersey}, vol.~2,
  no.~5, pp. 5--1, 1996.

\bibitem{godsil2013algebraic}
C.~Godsil and G.~F. Royle, \emph{Algebraic graph theory}.\hskip 1em plus 0.5em
  minus 0.4em\relax Springer Science \& Business Media, 2013, vol. 207.

\bibitem{rungger2016scots}
M.~Rungger and M.~Zamani, ``{SCOTS}: A tool for the synthesis of symbolic
  controllers,'' in \emph{Proceedings of the 19th International Conference on
  Hybrid Systems: Computation and Control}.\hskip 1em plus 0.5em minus
  0.4em\relax ACM, 2016, pp. 99--104.

\bibitem{katoen08}
C.~Baier and J.~P. Katoen, \emph{Principles of model checking}.\hskip 1em plus
  0.5em minus 0.4em\relax The MIT Press, 2008.

\end{thebibliography}
\end{document}